\documentclass[11pt,reqno]{amsart}
\usepackage{style}
\usepackage[foot]{amsaddr}
\usepackage[margin=1in]{geometry}

\title[Spectral algorithms optimally recover planted sub-structures]
{Spectral algorithms optimally recover \\
planted sub-structures}

\author{Souvik Dhara{$^\star$}, Julia Gaudio{$^\dagger$}, Elchanan Mossel{$^\ddagger$}, Colin Sandon}
\address[{$^\star$}]{The Simons Institute for the Theory of Computing, University of California, Berkeley}
\address[{$^\dagger$}]{Department of Industrial Engineering and Management Sciences, Northwestern University}
\address[{$^\ddagger$}]{Department of Mathematics, Massachusetts Institute of Technology}
\email{\href{mailto:dharasouvik1991@gmail.com}{dharasouvik1991@gmail.com},\href{mailto:julia.gaudio@northwestern.edu}{julia.gaudio@northwestern.edu}, \href{mailto:elmos@mit.edu}{elmos@mit.edu}, \href{mailto:csandon@comcast.net}{csandon@comcast.net}}
\thanks{\emph{Acknowledgement:} 
S.D., E.M., and C.S. were partially supported by Vannevar Bush Faculty Fellowship ONR-N00014-20-1-2826. S.D.~was supported by Simons-Berkeley Research Fellowship.
E.M. and C.S. were partially supported by NSF award DMS-1737944. E.M. was partially supported by Simons Investigator award (622132) and by ARO MURI W911NF1910217. J.G. was partially supported by NSF award CCF-2154100. Part of this work was completed while S.D and C.S. were at the MIT Mathematics Department.
} 
\begin{document}
\maketitle
\begin{abstract}

Spectral algorithms are an important building block in machine learning and graph algorithms. We are interested in studying when such algorithms can be applied directly to provide optimal solutions to inference tasks. Previous works by Abbe, Fan, Wang and Zhong (2020) and by Dhara, Gaudio, Mossel and Sandon (2022) showed the optimality for community detection in the Stochastic Block Model (SBM), as well as in a censored variant of the SBM.
Here we show that this optimality is somewhat universal as it carries over to other planted substructures such as the planted dense subgraph problem and submatrix localization problem, as well as to a censored version of the planted dense subgraph problem.
 \end{abstract}

\section{Introduction} 
Spectral algorithm are an important building block in machine learning and graph algorithms. We are interested in studying when such algorithms can be applied directly to provide optimal solutions. We study 
the applicability of such algorithm to the problem of finding dense matrix in a large structure. 
The problem of finding dense structure in a large matrix is an important machine learning task on network data. 
A canonical probabilistic  formulation of this problem is made by ``planting’’ a dense substructure hidden under random noise. More precisely, consider a symmetric $n\times n$ matrix $X$, and suppose there is an unknown subset $S^\star \subset [n]$ such that $X_{ij} \sim P$ if $i,j\in S^\star$ and $X_{ij} \sim Q$ otherwise, where $P$ and $Q$ are two distributions with different means.
The goal is to recover $S^\star$ based on an observation of $X$.
If $P,Q$ are Bernoulli distributions, the corresponding problem is referred to as the  {\em planted dense subgraph problem }(\PDS), which is a generalization of the well-known planted clique problem \cite{Kuc95,AKS98}. The case where $P,Q$ are Normal distributions is known as the \emph{submatrix localization problem} (\SL) \cite{SWPN09,KBRS11} in the high-dimensional statistics literature. 
The central questions in these areas are: (1) designing efficient algorithms for recovering the planted substructure, (2) finding parameter regimes when the recovery would be information theoretically impossible, as well as (3) showing parameter regimes where the recovery would be information theoretically possible, but known algorithms would fail (i.e. there is a computational versus statistical gap).
We refer the reader to the survey of Wu and Xu~\cite{WX21} for a detailed account on the recent developments on this topic and further references.

If the size of the planted subset is $|\true| = \Theta(n)$, then there is no computational versus statistical gap, as shown by Hajek, Wu, and Xu~\cite{Hajek2016, Hajek2016-hidden-community, Hajek2015c,Hajek2017}. We will focus on this regime throughout the paper. In this case, semidefinite programming (\textsc{SDP}) relaxation of the Maximum-Likelihood estimator \cite{Hajek2016, Hajek2016-hidden-community} recovers the hidden subset $\true$ in both the PDS and SL problems, up to the information-theoretic thresholds. An algorithm based on degree-thresholding and voting is also optimal for exact recovery in the PDS model \cite[Appendix A]{Hajek2015c};  message-passing \cite{Hajek2017} is also optimal for the SL model.  

Our main interest in this paper is in studying the following two aspects of these planted recovery problems:
\begin{itemize}
    \item We are interested in the power of spectral algorithms, which are simple, efficient, and widely used.
    \item We consider a {\em censored} version of the planted dense subgraph problem.
    Here the statuses of some edges are unknown as is often the case in real network inference problems. 
\end{itemize}
We note that some applications of spectral algorithms to the exact recovery problem use an additional combinatorial clean-up stage (see e.g.~\cite{CO16,Vu18,YP14}), but we follow \cite{Abbe2020,Dhara2021} in studying spectral algorithms without a clean-up stage.
This is partially motivated by the fact that most real applications of spectral algorithms do not include a combinatorial clean-up stage. Moreover, spectral algorithms are highly efficient; spectral decomposition can be computed in $O(n^{\omega})$ time~\cite{Banks21}, and in fact, finding only the top eigenvectors using power-iteration on sparse random matrices takes only $O(n\log^2n)$ time; see Remark~\ref{rem:run_time}.

This gives rise to the following main questions that we address in this paper: 
\begin{quote}
    \emph{Are there simple spectral algorithms that are optimal for recovering planted substructures? If so, can we design optimal algorithms when additionally there is missing data? }
\end{quote}

\noindent {\bf Our Contributions.} 
\begin{enumerate}
    \item We first study spectral algorithms for \PDS~and~\SL~models in the set up of Hajek, Wu and Xu~\cite{Hajek2016, Hajek2016-hidden-community}. We show for the SL model, a simple spectral algorithm based on thresholding the top eigenvector of the underlying random matrix $X$ succeeds in exact recovery up to the information theoretic threshold (Theorem \ref{theorem:submatrix-localization}). In the PDS model, a similar algorithm based on thresholding a suitable linear combination of the top two eigenvectors of the random matrix is optimal (Theorem~\ref{theorem:dense-subgraph}). 
    \item Next we consider the recovery problem in the censored version when the information about the edge statuses are \emph{missing at random} (see Definition~\ref{defn:CPDS}). To the best of our knowledge, the \PDS~model with missing data had not been studied in the literature. We obtain the information theoretic threshold in the censored set up (see Theorem~\ref{theorem:impossibility2}). The spectral algorithm based on $X$ is not always optimal. 
    To design an efficient spectral algorithm, we consider an operator called the \emph{signed adjacency matrix} (see \eqref{defn:signed-adj-matrix}), and show that a linear combination of the top two eigenvectors of this signed adjacency matrix can recover the \PDS~up to the information theoretic threshold (see Theorem~\ref{theorem:exact-recovery-censored}).
    Our results focus on the parameter regimes where there are no computational versus statistical gaps, and handle more general censoring regimes than what was considered in prior work on spectral algorithms for community detection \cite{Dhara2021}.
\end{enumerate}
Our proofs follow a similar pattern for all models considered. Indeed our results and proofs show that the optimality of spectral algorithms is somewhat universal. 

\subsection{Model description and objective}
We start by describing the models for the planted subgraph recovery problem, a version of this model with missing data, and the submatrix localization problem. 
\begin{definition}[$\PDS$ and Censored $\PDS$ model] \label{defn:CPDS}
In the {\em Planted Dense Subgraph Problem}, there are $n$ vertices, labeled $\{1, 2, \dots, n\}$ and $\true$ is drawn uniformly at random from all size-$K$ subsets of $[n]$.
A pair of vertices $\{i,j\}$ is connected with probability $p$ if $i, j \in \true$. Otherwise, the pair is connected with a different probability $q$. 
We refer to this as the \PDS~model, denoted by $\PDS(p,q,K)$. 

In the {\em Censored Planted Dense Subgraph Problem}, we additionally have each edge status being revealed to us independently with probability $\alpha$. In this case, the output is a graph with edge statuses given by $\{\texttt{present, absent, censored}\}$. We denote this model by $\CPDS(p,q,\alpha,K)$. 
\end{definition}

\begin{definition}[$\SL$ model] 
In the  \emph{Submatrix Localization (SL) Problem}, we have an $n \times n$ symmetric matrix~$A$. As before, $\true$ is drawn uniformly at random from all size-$K$ subsets of $[n]$. The entries of $(A_{ij})_{i\leq j}$ are independent, and $A_{ij} \sim \textsc{Normal}(\mu, 1)$ if $i, j \in \true$, and $A_{ij} \sim \textsc{Normal}(0, 1)$ otherwise.
Throughout, we assume that $\mu>0$.
We denote this model by $A \sim \text{SL}(\mu, K)$.
\end{definition}

\noindent {\bf Objective. } 
Suppose $\true$ is unknown.
We want to recover $\true$ \emph{exactly}, i.e., we want to find an estimator $\hat{S}_n \subset [n]$ such that 
\begin{eq} \label{exact-recovery-possible}
\lim_{n\to\infty} \PR(\hat{S}_n  = \true) = 1. 
\end{eq}
Throughout, we will assume that the parameters such as $K, p, q, \alpha, \mu$ are known. 

\subsection{Main results} 
\subsubsection*{The $\PDS$ problem.}

We will first consider spectral recovery on the \PDS~model for the regime where
\begin{eq}\label{eq:PDS-parameters}
   p = \frac{a\log n}{n}, \quad  q = \frac{b\log n}{n}, \quad K = \lfloor \rho n\rfloor, \quad a,b,\rho \text{ are fixed constants}.
\end{eq}
To state the result, define 
\begin{eq}\label{eq:defn:f}
   f(a,b) := a - \left(\frac{a-b}{\log a - \log b} \right) \log \left(\frac{\e a(\log a - \log b)}{a-b} \right), 
\end{eq}
for $a,b > 0$ such that $a \neq b$. Additionally, (for $a,b \geq 0$) define $f(a,0) = a$, $f(0,b) = b$, and $f(a,a) = 0$. The function $f(a,b)$ is symmetric in its arguments.
Our first main result is the following.
\begin{theorem}\label{theorem:dense-subgraph-1}
Consider the \PDS~model with parameters given by \eqref{eq:PDS-parameters}. If $\rho f(a,b) > 1$, then there is a spectral algorithm that outputs $\hat{S}_n$ such that 
$\lim_{n\to\infty}\PR(\hat{S}_n = \true) =1$.
\end{theorem}
The spectral algorithm is based on a linear combination of the top two eigenvectors of the adjacency matrix as described in Algorithm~\ref{alg:PDS}, and in particular, does not require additional combinatorial  clean-up steps. 
Hajek et~al.~\cite{Hajek2016} showed that the recovery of $\true$ is information theoretically impossible if $\rho f(a,b)<1$. Therefore, Theorem~\ref{theorem:dense-subgraph} shows that a spectral algorithm works up to the information theoretic threshold.  The achievability and impossibility parameter regimes in the PDS model are summarized by Figure \ref{fig:PDS-regimes}.
\begin{figure}[h]
    \centering
    \includegraphics[scale=0.35]{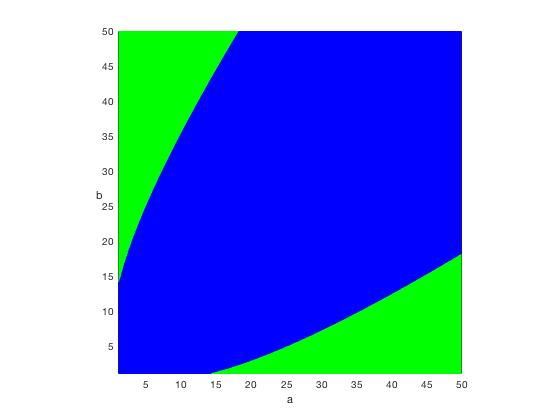}
    \caption{Phase diagram for the PDS model, for $\rho = \frac{1}{4}$. In the green region, estimation of $\true$ is possible by a spectral algorithm, while in the blue region, it is information-theoretically impossible.}
    \label{fig:PDS-regimes}
\end{figure}

\subsubsection*{The $\SL$ problem. } Next, we state the results about spectral recovery for the $\SL$ model when
\begin{align} \label{eq:SL-parameters}
    \mu = a \sqrt{\frac{\log n}{n}}, \quad K = \lfloor \rho n\rfloor, \quad a,\rho \text{ are fixed constants.}
\end{align}

\begin{theorem}\label{theorem:submatrix-localization-2}
Consider the $\SL$ model with parameters given by \eqref{eq:SL-parameters}, and assume that $\rho a^2 > 8$. Then there is a spectral algorithm that outputs $\hat{S}_n$ such that $\lim_{n\to\infty}\PR(\hat{S}_n = \true) =1$.
\end{theorem}
The spectral algorithm for the $\SL$~model only needs to consider the top eigenvector of the underlying matrix (see Algorithm~\ref{alg:Gaussian}) as compared to  top two eigenvectors for the $\PDS$~model. The fact that exact recovery of $\true$ is information theoretically impossible for $\rho a^2 < 8$ was proven by Hajek et~al.~\cite{Hajek2016-hidden-community}. 

\subsubsection*{The $\CPDS$ problem.}
Next, we discuss multiple results about the censored model. To the best of our knowledge, this model was not studied in the prior literature.
Interestingly, if $p,q = \Omega(1)$, then  the \emph{na\"ive} spectral algorithm based on the adjacency matrix with missing entries replaced by 0 does not lead to an optimal algorithm.
Instead, we have to consider a spectral algorithm using a ternary encoding matrix called the \emph{signed adjacency matrix} that encodes information about the labels $\{\texttt{present, absent, censored}\}$. However, the above na\"ive spectral algorithm turns out to be optimal when $p,q = o(1)$. 

The results for the $\CPDS$ model will be stated in terms of a Chernoff--Hellinger divergence, introduced by Abbe and Sandon~\cite{Abbe2015}.  
\begin{definition}[Chernoff--Hellinger Divergence]
Given two vectors $\mu, \nu \in (\R_+\setminus\{0\})^k$,
define
\[D_x(\mu, \nu) = \sum_{i \in[k]}  \big[x \mu_i  + (1-x) \nu_i - \mu_i^x \nu_i^{1-x}\big] \quad \text{ for }x \in [0,1].\]
The Chernoff--Hellinger divergence of $\mu$ and $\nu$ is defined as
\begin{eq} \label{eq:CH-defn}
   \Delta_+(\mu, \nu) = \max_{x \in [0,1]} D_x(\mu, \nu). 
\end{eq}
Given $p,q\in [0,1]$, we simply write $\Delta_+(p,q)$ for $\Delta_+((p,1-p),(q,1-q))$. We also define $\Delta_+(p,q)$ when $|\{p, q\} \cap \{0,1\} | \geq 1$ by taking the continuous extension, which yields $\Delta_+(p,q) = |p-q|$ in these cases.
\end{definition}

The first main result for the $\CPDS$ model is the following. 
\begin{theorem}\label{theorem:exact-recovery-censored-1}
Consider the $\CPDS(p,q,\alpha,K_n)$ model from Definition~\ref{defn:CPDS} with $K_n = \lfloor \rho n\rfloor $, $\alpha_n = \nicefrac{t\log n}{n}$ for some fixed $\rho, p, q \in (0,1)$ and $t>0$.
Suppose that $t\rho \Delta_+(p,q) > 1$. Then there is a spectral algorithm whose output $\hat{S}_n$ satisfies $\lim_{n\to\infty}\PR(\hat{S}_n = \true) =1$.  
\end{theorem} 
As remarked earlier, the spectral algorithm (Algorithm~\ref{alg:CPDS}) in this case uses a signed adjacency matrix. 
The next result finds the information theoretic threshold for the \CPDS~model.
\begin{theorem}\label{theorem:impossibility2} 
Consider the $\CPDS(p,q,\alpha,K_n)$ model from Definition~\ref{defn:CPDS} with $K_n = \lfloor \rho n\rfloor $ for some fixed $\rho\in (0,1)$. 
Suppose that $\min(p_n,q_n,1-p_n,1-q_n)=n^{-o(1)}$, 
and 
\begin{eq}\label{impossibility-condition}
   \limsup_{n\to\infty} \frac{\alpha_n n}{\log n} \times \rho \Delta_+(p_n,q_n) <  1.
\end{eq}
Given the unlabeled observation $G $ from $\CPDS(p,q,K,\alpha)$, we have that any estimator $\hat{S}_n$ satisfies $\lim_{n\to\infty} \PR(\hat{S}_n = \true) = 0$.  
\end{theorem}
The achievability and impossibility parameter regimes in the CPDS model are summarized by Figure \ref{fig:CPDS-regimes}, in the case where $p$, $q$, $t$, and $\rho$ are constants.
\begin{figure}
    \centering
    \includegraphics[scale=0.35]{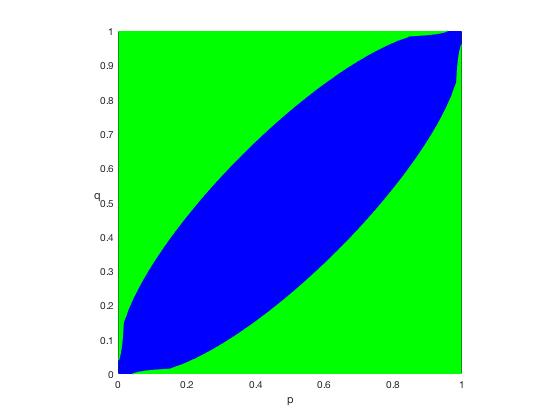}
        \caption{Phase diagram for the CPDS model, for $\rho = \frac{1}{4}$ and $t  = 100$. In the green region, estimation of $\true$ is possible by a spectral algorithm, while in the blue region, it is information-theoretically impossible.}
    \label{fig:CPDS-regimes}
\end{figure}
Next, we relate the achievability conditions for the \PDS~and the  \CPDS~problems. In particular, we prove the following.
\begin{lemma}\label{lemma:connected-conditions}
Let $a, b > 0$ be constants.
Suppose $\alpha_n = \frac{t_n \log n}{n}$, $\lim_{n\to\infty}p_n t_n = a$, and $\lim_{n\to\infty}q_n t_n = b$, where $p_n, q_n = o(1)$. 
Then $\lim_{n\to\infty}t_n \Delta_+(p_n,q_n) = f(a,b)$. 
\end{lemma}
In light of this lemma, we now have the following result for the \CPDS~problem that exhibits the identical
recovery regime when $p_n,q_n = o(1)$ as compared to the \PDS~problem in Theorem~\ref{theorem:dense-subgraph-1}.
\begin{corollary}\label{cor-CPDS-sparse-PDS-1}
Consider the $\CPDS(p_n,q_n,\alpha_n,K_n)$ model from Definition~\ref{defn:CPDS} with $K_n = \lfloor \rho n\rfloor $, where $\rho\in (0,1)$ is fixed. Let $a,b>0$ be constants  satisfying $\alpha_n p_n = \frac{a \log n}{n}, \alpha_n q_n = \frac{b \log n}{n}$.
Moreover, let $p_n,q_n = o(1)$.
If $ \rho f(a,b) >  1$, then there is a spectral algorithm that recovers $\true$ with probability $1-o(1)$.
\end{corollary}
\begin{proof}
Let $A'$ be the matrix where $A_{ij}' = 1$ if the edge $\{i,j\}$ is present and $A_{ij}' = 0$ otherwise. That is, we encode both the censored and absent edges by $0$. Observe that $A'$ is distributed as the adjacency matrix of $\PDS\big(\frac{a \log n}{n}, \frac{b \log n}{n}, K_n\big)$. The statement then follows from Theorem \ref{theorem:dense-subgraph-1}.
\end{proof}

Finally, we consider the optimality of spectral algorithms under a more general parameter regime.
\begin{theorem}\label{theorem:sparse-case-sharpness}
Consider the $\CPDS(p_n,q_n,\alpha_n,K_n)$ model from Definition~\ref{defn:CPDS} with $K_n = \lfloor \rho n\rfloor $, where $\rho\in (0,1)$ is fixed. Let $a,b>0$ be constants  satisfying $\alpha_n p_n = \frac{a \log n}{n}, \alpha_n q_n = \frac{b \log n}{n}$. Then the following hold: 
\begin{enumerate}
    \item If $ \liminf_{n\to\infty} \frac{\alpha_n n}{\log n} \times \rho \Delta_+(p_n,q_n) >  1$, then there is a spectral
    algorithm that recovers $\true$ with probability $1-o(1)$.
    \item If $ \limsup_{n\to\infty} \frac{\alpha_n n}{\log n} \times \rho \Delta_+(p_n,q_n) <  1$ and $\min(1-p_n,1-q_n)=n^{-o(1)}$,  then every algorithm fails to recover $\true$ with probability $1-o(1)$.
\end{enumerate}
\end{theorem}
The following result follows directly from Theorem \ref{theorem:sparse-case-sharpness}, complementing Corollary \ref{cor-CPDS-sparse-PDS-1}.

\begin{corollary}\label{cor-CPDS-sparse-PDS-2}
Under the assumptions of Corollary \ref{cor-CPDS-sparse-PDS-1}, 
if $ \rho f(a,b) <  1$, then every algorithm fails to recover $\true$ with probability $1-o(1)$.
\end{corollary}

\subsection{Proof ideas} 
We first give some intuition behind the success of spectral algorithms for estimating~$\true$ for the $\PDS$ problem. The ideas for the $\SL$ problem are similar in spirit.
Let $A^{\star} = \mathbb{E}[A]$ be the expected value of the adjacency matrix. Then $A^{\star}$ is a rank-$2$ matrix with two non-zero eigenpairs $(\lambda_1^{\star}, u_1^{\star})$ and $(\lambda_2^{\star}, u_2^{\star})$, where $\lambda_1^\star \geq \lambda_2^\star$. Both eigenvectors are block eigenvectors, taking on one value on entries corresponding to $\true$, and another value on entries corresponding to $[n] \setminus \true$. Given the matrix $A$, we can compute its eigenpairs $(\lambda_1, u_1), (\lambda_2, u_2), \dots$, where $\lambda_1 \geq \lambda_2 \geq \cdots$. Drawing on eigenvector perturbation theory, one might expect $u_1,u_2$ to qualitatively behave like $u_1^{\star},u_2^{\star}$. However, for the exact recovery problem, one needs a strong bound on the $\ell_{\infty}$ norm,  referred to as an \emph{entrywise} perturbation bound. That is, if it were the case that $\Vert u_1 - u_1^{\star}\Vert_{\infty}$ and $\Vert u_2 - u_2^{\star}\Vert_{\infty}$ were small, then we could threshold the values of $u_1$ or $u_2$ in order to accurately estimate $\true$.

Unfortunately, $u_1$ and $u_2$ are not well enough approximated by $u_1^{\star}$ and $u_2^{\star}$ in an entrywise sense. On the other hand, work by Abbe, Fan, Wang and Zhong~\cite{Abbe2020} identifies conditions under which an eigenvector $u_i$ of a random matrix $A$ is well-approximated by a different vector, namely the ratio $\nicefrac{A u_i^{\star}}{\lambda_i^{\star}}$. Such a ratio is often straightforward to analyze; in community detection scenarios, each entry of $\nicefrac{A u_i^{\star}}{\lambda_i^{\star}}$ is a weighted summation of independent Bernoulli random variables. Abbe et~al.~\cite{Abbe2020} applied their entrywise eigenvector analysis technique in order to show optimality of a simple spectral algorithm for community detection in the SBM (see \cite[Theorem 3.2]{Abbe2020}), in addition to other estimation problems involving matrices ($\mathbb{Z}_2$-synchronization and noisy matrix completion). Recently, Dhara, Gaudio, Mossel and Sandon~\cite{Dhara2021} applied the entrywise eigenvector analysis technique to show that a spectral algorithm is optimal for the community detection problem in a censored symmetric stochastic block model, where one has two communities of approximately equal sizes and the edge densities inside the two communities are equal.

One important technical distinction in our work compared to the community detection literature is that in community detection, we seek to determine the communities up to a global flip of the community labels. On the other hand, in the PDS model, we need to specifically identify the vertices within the planted dense subgraph. The entrywise eigenvector analysis technique allows us to bound
\[\min_{s\in \{-1,1\}} \left\Vert su_i - \frac{A u_i^{\star}}{\lambda_i^{\star}}\right \Vert_{\infty}.\]
That is, we can show that $u_i$ or $-u_i$ is close to $ \nicefrac{A u_i^{\star}}{\lambda_i^{\star}}$ in an entrywise sense, but the technique gives us no way to dissambiguate the two options. Moreover, we use a weighted combination $u_1$ and $u_2$ in the PDS algorithm, with weights $c_1^{\star}$ and $c_2^{\star}$. We are able to show that 
\[\min_{s_1, s_2\in \{-1,1\}} \left\Vert s_1 c_1^{\star} u_1 + s_2 c_2^{\star} u_2- A \left(\frac{c_1^{\star} u_1^{\star}}{\lambda_1^{\star}} + \frac{c_2^{\star} u_2^{\star}}{\lambda_2^{\star}} \right)\right \Vert_{\infty}\]
is small, so that the recovery performance of $s_1 c_1^{\star} u_1 + s_2 c_2^{\star} u_2$ matches the recovery performance of $A \left(\frac{c_1^{\star} u_1^{\star}}{\lambda_1^{\star}} + \frac{c_2^{\star} u_2^{\star}}{\lambda_2^{\star}} \right)$, for those $s_1, s_2$ which minimize the above expression. It remains to determine $s_1$ and $s_2$, which amounts to determining the orientations of $u_1$ and $u_2$. Given two choices for $s_1$ and $s_2$, we obtain four candidate linear combinations, and therefore four candidates estimates of $\true$. In order to determine which one is correct, we choose the one that maximizes the likelihood of the observed adjacancy matrix $A$. Since choosing the correct orientation amounts to maximizing the likelihood over a restricted number of candidates, we determine the regime where the Maximum Likelihood estimator succeeds in recovering the communities. In this regime, it suffices to show that $\true$ is one of the four candidates, with high probability.

The weights $c_1^{\star}$ and $c_2^{\star}$ cannot be arbitrary, but must be chosen correctly in order to achieve our optimal results. Using the insight that $c_1^{\star} u_1 + c_2^{\star} u_2$ (in the correct orientation) behaves like $v := A \left(\frac{c_1^{\star} u_1^{\star}}{\lambda_1^{\star}} + \frac{c_2^{\star} u_2^{\star}}{\lambda_2^{\star}} \right)$, we choose the weights so that the vector $v$ optimally recovers the planted dense subgraph. In particular, we choose the weights so that $v_i = \frac{1}{\sqrt{K}}\sum_{j \in \true} A_{ij}$.

The spectral algorithm for the CPDS model is similar in spirit to the spectral algorithm for the PDS model. The key difference is that unlike in the PDS model, where each pair of vertices is associated with one of two possible observations ($\texttt{present, absent}$), each pair of vertices is associated with one of three possible observations ($\texttt{present, absent, censored}$). Therefore, the adjacency matrix must be encoded using three numeric values. The choice of encoding matrix in \eqref{defn:signed-adj-matrix} is such that it balances the relative information contributed by revealed edges and non-edges. 
Recently,~\cite{Dhara2021} studied spectral algorithms using such signed adjacency matrices for community detection on a censored version of the stochastic block model. It was shown that while the spectral algorithms using signed adjacency matrices are optimal when the within community edge probabilities are the same, it may be the case that these spectral algorithms are sub-optimal when the within community edge probabilities are different (cf.~\cite[Theorem 2.6]{Dhara2021}). 
Note that the \CPDS~model would correspond to a censored stochastic block model, where the within community edge probabilities are different, and moreover, the community sizes are asymptotically unequal. The results in this paper show that these spectral algorithms are still optimal for the recovery problem in~\CPDS~and in particular are in contrast with the counterexamples provided in \cite[Theorem 2.6]{Dhara2021}. 
Finally, our analysis of the spectral algorithm in the censored case requires us to additionally show that the Maximum Likelihood Estimator (MLE) achieves exact recovery up to the information-theoretic threshold (see Lemma~\ref{lemma:MLE}).

The proof for the impossibility of recovery follows by showing that the Maximum a Posteriori (MAP) estimator, which is an optimal estimator, would fail in recovering $\true$. The idea is closely related to the analysis of the Censored SBM in \cite{Dhara2021}. We first show that the MAP and MLE are equivalent, and then identify a condition under which the MLE fails to determine $\true$. The condition requires several pairs of vertices, such that if we swap their assignments, the likelihood of the observed graph does not change. If there are enough such pairs, then the MLE is likely to fail in determining~$\strue$.

\subsection{Discussion and future work}
Our work leaves several directions for future work.
\begin{enumerate}
    \item 
    We focus on the case where $K = \Theta(n)$. While the entrywise eigenvector analysis method of~\cite{Abbe2020} allows us to handle slightly sublinear $K$, it does not allow us to match existing results for SDPs in the PDS and SL problems. When $K = \omega\left(\nicefrac{n}{\log n} \right)$, the SDP threshold matches the information-theoretic threshold with sharp constants \cite{Hajek2016-hidden-community}. On the other hand, when $K = o\left(\nicefrac{n}{\log n}\right)$, then the SDP is order-wise suboptimal \cite{Hajek2016-hidden-community}. Finally, if $K = \Theta\left(\nicefrac{n}{\log n}\right)$, then the SDP is suboptimal by a constant factor, though order-wise optimal \cite{Hajek2016-hidden-community}. It was conjectured that a spectral algorithm would require a stronger signal than the SDP algorithm. Comparing the performance of SDPs and spectral algorithms for sublinear $K$ is an interesting avenue for future work.
    \item 
    In the CPDS model, can we consider a more general censoring distribution that can model the case where the edges' statuses are not missing at random? 
\end{enumerate}

\noindent {\bf Organization.}
The remainder of the paper is structured as follows. 
We start by setting up some preliminary notation below. 
In Section~\ref{sec:algs} we provide the main algorithms and state the main algorithmic results. 
The analyses of the spectral algorithms for \PDS~and~\CPDS ~are provided respectively in Sections~\ref{sec:PDS}~and~\ref{sec:CPDS}. We  prove the impossibility result for \CPDS~in Section~\ref{sec:impossible}, in addition to demonstrating the tightness of our results under a more general parameter regime. The spectral recovery in the \SL~problem is analyzed in Section~\ref{sec:Gaussian}.
\\

\noindent {\bf Notation.} For a vector $x \in \mathbb{R}^n$, we define $\Vert x \Vert_2 = (\sum_{i=1}^n x_i^2)^{1/2}$ and $\Vert x \Vert_{\infty} = \max_i |x_i|$. For a matrix $M \in \mathbb{R}^{n \times d}$, we use $M_{i \cdot}$ to refer to its $i$-th row, represented as a row vector. 
Given a matrix $M$, $\Vert M \Vert_2 = \max_{\Vert x \Vert_2 = 1} \Vert M x \Vert_2$ is the spectral norm, 
and $\Vert M \Vert_{2 \to \infty} = \max_i \Vert M_{i \cdot}\Vert_2$ is the matrix $2 \to \infty$ norm. 
We use the convention that $\log$ denotes natural logarithm, and write $\log^k(n)$ to mean $\left(\log(n)\right)^k$. Finally, given two sequences $\{s_n\}_{n=1}^{\infty}$ and $\{r_n\}_{n=1}^{\infty}$, we write $s_n \asymp r_n$ when $\lim_{n\to \infty} \frac{s_n}{r_n} = 1$.

\section{Main Algorithms} \label{sec:algs} 
\subsection{Spectral recovery of \PDS} 
We will first consider spectral recovery in the \PDS~model for the regime when the parameters are given by \eqref{eq:PDS-parameters}.
The algorithm uses a linear combination of the top two eigenvectors of the adjacency matrix. Let us start by describing how to calculate the coefficients of this linear combination:\\

\begin{breakablealgorithm}
\caption{Calculation of weights}\label{alg:weights}
\begin{algorithmic}[1]
\Require{An $n \times n$ matrix $B$ of rank $2$, and a $K\in \N$.}
\vspace{.2cm}
\Ensure{A pair of weights.}
\State Find the top two eigenpairs of $B$, denoting them $(\gamma_1, w_1)$ and $(\gamma_1, w_2)$, where $\gamma_1 \geq \gamma_2$.
\State Letting $\vecS$ be the vector whose first $K$ entries are $1/\sqrt{K}$ and whose other entries are $0$, let $c_1^{\star}$ and $c_2^{\star}$ be such that $$\frac{c_1^{\star} \log(n)}{\gamma_1} w_1 + \frac{c_2^{\star} \log(n)}{\gamma_2} w_2 = \vecS.$$
\State Return $(c_1^{\star}, c_2^{\star})$.
\end{algorithmic}
\end{breakablealgorithm}
\vspace{.6cm}
Based on the weights above, we now describe the spectral algorithm: \\
\begin{breakablealgorithm}
\caption{Spectral recovery in the PDS model}\label{alg:PDS}
\begin{algorithmic}[1]
\Require{Parameters $a$, $b$, and $K$; an adjacency matrix $A$ of $\PDS$ model with parameters in \eqref{eq:PDS-parameters}.}
\vspace{.2cm}
\Ensure{An estimate of the planted dense subgraph vertices.}
\State Let $B$ be the matrix where $B_{ij} = \frac{a \log n}{n}$ when $i, j \leq K$, and $B_{ij} = \frac{b \log n}{n}$ otherwise. \label{step:PDS1}
\State Apply Algorithm \ref{alg:weights} to $(B, K)$, obtaining the weights $(c_1^{\star}, c_2^{\star})$.\label{step:PDS2}
\State Compute the top two eigenpairs of $A$, denoting them $(\lambda_1, u_1)$ and $(\lambda_2, u_2)$, where $\lambda_1 \geq \lambda_2 \geq \cdots$. \label{step:PDS3}
\State Let $U = \{s_1 c_1^{\star} u_1 + s_2 c_2^{\star} u_2 : s_1, s_2 \in \{-1,1\}\}$. \label{step:PDS4}
\State For each $u \in U$ identify $\hat{S}(u)$ as the set of vertices with the $K$ largest entries in the vector $u$.
\label{step:PDS5}
\State Return $\hat{S}(u)$ which maximizes the likelihood $\mathbb{P}(A | S^{\star} = \hat{S}(u))$, over $u \in U$. \label{step:PDS6}
\end{algorithmic}
\end{breakablealgorithm}
\vspace{.6cm}
To state the result, recall the definition of $f$ from \eqref{eq:defn:f}.
The following result implies Theorem \ref{theorem:dense-subgraph-1}.
\begin{theorem}\label{theorem:dense-subgraph}
Consider the \PDS~model with parameters given by \eqref{eq:PDS-parameters}, and let $\hat{S}_n$ be the output of Algorithm \ref{alg:PDS}. If $\rho f(a,b) > 1$, then $\lim_{n\to\infty}\PR(\hat{S}_n = \true) =1$.
\end{theorem}
As remarked earlier, the recovery of $\true$ is information-theoretically impossible if $\rho f(a,b)<1$~as shown in~\cite{Hajek2016}. Hence, Theorem~\ref{theorem:dense-subgraph} proves optimality of Algorithm~\ref{alg:PDS}.

\begin{remark} \label{rem:run_time} 
To analyze the runtime of Algorithm \ref{alg:PDS}, note that Steps \ref{step:PDS1} and \ref{step:PDS2} are not data-dependent, and can therefore be precomputed. The runtime then depends on Steps \ref{step:PDS3}, \ref{step:PDS5}, and \ref{step:PDS6}. In order to compute the eigenpairs in Step \ref{step:PDS3}, one can use the power method on the matrix $A$.
Letting $\lambda_1^{\star}$ and $\lambda_2^{\star}$ be the eigenvalues of $\mathbb{E}[A]$, Weyl's inequality states that  $|\lambda_1 -\lambda_1^{\star}|, |\lambda_2 -\lambda_2^{\star}| \leq \Vert A - \mathbb{E}[A] \Vert_2$. In Lemma \ref{lemma:spectral-norm-concentration}, we show that $\Vert A - \mathbb{E}[A] \Vert_2 = O(\sqrt{\log n})$ with high probability. On the other hand, we will show that $\lambda_1^{\star}, \lambda_2^{\star} = \Theta(\log n)$. Therefore, the spectral gap of $A$, which is $\delta \triangleq \frac{\lambda_1 - \lambda_2}{\lambda_1}$, is $\Theta(1)$. Since the power method converges in $O\left(\log(n)/\delta \right)$ iterations (see e.g. \cite{Garber2016}), and each iteration requires $O(n \log n)$ time for multiplication of $A$ by a vector, the overall runtime of computing the top eigenvector of $A$ is $O(n \log^2 (n))$. To obtain the second eigenvector, we can deflate $A$ by substracting $\lambda_1 u_1 u_1^T$. By similar analysis, the power method will require $O(n \log^2 (n))$ time to obtain the second eigenvector. Next, Step \ref{step:PDS5} requires sorting four vectors of length $n$, which can be done in $O(n \log n)$ time. Step \ref{step:PDS6} requires counting the number of edges within each of the four estimated dense subgraphs. Given a sparse representation of $A$, this can be done in time $O(n \log n)$. Therefore, the overall runtime of Algorithm \ref{alg:PDS} is $O(n \log^2 (n))$, significantly faster than what would be possible with standard SDP algorithms.
\end{remark}

\subsection{Spectral recovery for $\SL$.}
Next, we consider a spectral algorithm based on only the top eigenvector of the underlying matrix, which is shown to be optimal. 
\begin{breakablealgorithm}
\caption{Spectral submatrix localization}\label{alg:Gaussian}
\begin{algorithmic}[1]
\Require{Parameters $\mu$ and $K$; a matrix $A \sim \SL(\mu, K)$} \vspace{.2cm}
\Ensure{An estimate of $\true$}
\State Find the top eigenvector $u$ of $A$. 
\State Identify $\hat{S}_+(u)$ as the set of vertices with the $K$ largest entries in the vector $u$. Similarly, identify $\hat{S}_-(u)$ as the set of vertices with the $K$ smallest entries in the vector $u$.
\State Return $\hat{S}(u) \in \left\{\hat{S}_+(u), \hat{S}_-(u)\right\}$ which maximizes the likelihood $\mathbb{P}(A | S^{\star} = \hat{S}(u))$.
\end{algorithmic}
\end{breakablealgorithm}
\vspace{.6cm}
The following theorem implies Theorem~\ref{theorem:submatrix-localization-2}.
\begin{theorem}\label{theorem:submatrix-localization}
Consider the $\SL$ model with parameters given by \eqref{eq:SL-parameters}.
Let $\hat{S}_n$ be the output of Algorithm \ref{alg:Gaussian}. If $\rho a^2 > 8$, then $\lim_{n\to\infty}\PR(\hat{S}_n = \true) =1$.
\end{theorem}
Using the fact from~\cite{Hajek2016-hidden-community} that exact recovery of $\true$ is information theoretically impossible for $\rho a^2 < 8$ (Corollary \ref{corollary:Gaussian-threshold}), Theorem~\ref{theorem:submatrix-localization} establishes optimality of Algorithm~\ref{alg:Gaussian}.
In order to prove Theorem \ref{theorem:submatrix-localization}, we will use an entrywise eigenvector perturbation result akin to the proof of Theorem~\ref{theorem:dense-subgraph}.

\subsection{Spectral recovery for \CPDS} 
We next describe a simple spectral algorithm that can recover the planted dense subgraph in the censored model up to the information theoretic boundary. 
Our spectral algorithm uses a \emph{signed adjacency matrix} representation, denoted by $A$. 
We define 
\begin{eq} \label{defn:signed-adj-matrix}
A_{ij} &= \begin{cases}
1 &\text{if } \{i,j\} \text{ is present}\\
-y &\text{if } \{i,j\} \text{ is absent}\\
0 &\text{if } \{i,j\} \text{ is censored}
\end{cases}
\end{eq}
where
\begin{eq}\label{defn:y}
y(p,q) = \frac{\log \frac{1-q}{1-p} }{\log \frac{p}{q} }.
\end{eq}
for $0 < p,q < 1$.
Recently, Dhara, Gaudio, Mossel, and Sandon studied spectral algorithms on such signed adjacency matrices, to perform community detection on a censored version of the Stochastic Block Model \cite{Dhara2021}. 
Note that the \CPDS~model would correspond to a stochastic block model, where the edge probabilities inside community 1 is $p$, and the edge probabilities inside community 2 and between two communities are $q$. While \cite{Dhara2021} had shown that  spectral algorithms based on the top two eigenvectors of a signed adjacency matrix may not work up to the information theoretic threshold (cf.~\cite[Theorem 2.6]{Dhara2021}), we will show that these algorithms are optimal for the CPDS model.

Let us now describe the spectral algorithm for recovery in the CPDS model.\\

\begin{breakablealgorithm}
\caption{Spectral recovery in the CPDS model}\label{alg:CPDS}
\begin{algorithmic}[1]
\Require{Parameters $p$, $q$, $\alpha$, and $K$; a signed adjacency matrix $A$ computed from an unlabeled observation of  $\CPDS(p,q,\alpha,K)$} \vspace{.2cm}
\Ensure{An estimate of the planted dense subgraph vertices}
\State Let $B$ be the matrix where $B_{ij} = \alpha \left(p -y(p,q)(1-p) \right)$ for $i, j \leq K$, and $B_{ij} = \alpha \left(q -y(p,q)(1-q) \right)$ otherwise.
\State Apply Algorithm \ref{alg:weights} to $B$, obtaining the weights $(c_1^{\star}, c_2^{\star})$.
\State Compute the top two eigenpairs of $A$, denoting them $(\lambda_1, u_1)$ and $(\lambda_2, u_2)$, where $\lambda_1 \geq \lambda_2 \geq \cdots$.
\State Let $U = \{s_1 c_1^{\star} u_1 + s_2 c_2^{\star} u_2 : s_1, s_2 \in \{-1,1\}\}$.
\State For each $u \in U$ identify $\hat{S}(u)$ as the set of vertices with the $K$ largest entries in the vector $u$.
\State Return $\hat{S}(u)$ which maximizes the likelihood $\mathbb{P}(A | S^{\star} = \hat{S}(u))$, over $u \in U$.
\end{algorithmic}
\end{breakablealgorithm}
\vspace{.6cm}
One can show that the runtime of Algorithm \ref{alg:CPDS} is $O(n \log^2(n))$, by a similar analysis to that of Algorithm \ref{alg:PDS}.

The following result implies Theorem \ref{theorem:exact-recovery-censored-1}.
\begin{theorem}\label{theorem:exact-recovery-censored} Let $K = \lfloor \rho n \rfloor$, where $\rho \in (0,1)$ is fixed. Suppose $p,q$ are constants such that $p,q \in (0,1)$ and $\alpha = \nicefrac{t\log n}{n}$. Let $\hat{S}_n$ be the output of Algorithm~\ref{alg:CPDS}.  If $t\rho > \nicefrac{1}{\Delta_+(p,q)}$, then $\lim_{n\to\infty}\PR(\hat{S}_n = \true) =1$.  \end{theorem}

\section{Analyzing spectral algorithms for \PDS} \label{sec:PDS}
In this section, we analyze Algorithm~\ref{alg:PDS} and complete the proof of Theorem~\ref{theorem:dense-subgraph}. 
We apply the method of entrywise perturbation analysis of eigenvectors developed recently by
 Abbe,  Fan, Wang, and Zhong
\cite{Abbe2020}. We will show the following: 
\begin{lemma}\label{lemma:entrywise}
Let $\rho \in (0,1)$ and $a,b \geq 0$ be constants. Let $A \sim \PDS(p,q,K)$, with parameters given by \eqref{eq:PDS-parameters}. Let $c_1, c_2$ be constants. Then with probability $1- O(n^{-3})$, 
\begin{equation}
\min_{s_1, s_2 \in \{-1,1\}}
\left \Vert  s_1 c_1 u_1 + s_2 c_2 u_2 - \left(c_1 \frac{A u_1^{\star}}{\lambda_1^{\star}} + c_2 \frac{A u_2^{\star}}{\lambda_2^{\star}} \right)\right \Vert_{\infty} \leq \frac{C}{\log \log (n) \sqrt{n}}, \label{eq:entrywise-planted-dense} \end{equation}
where $C = C(a,b,\rho, c_1, c_2)$ is a constant depending on $a$, $b$, $\rho$, $c_1$, and $c_2$.
\end{lemma}
Given Lemma~\ref{lemma:entrywise}, we can analyze the output of Algorithm \ref{alg:PDS} by analyzing the vector $A\vecS$, where $A\vecS$ is described in Algorithm \ref{alg:weights}.

\begin{proof}[Proof of Theorem \ref{theorem:dense-subgraph}]
~Since we are in the achievable regime of the PDS model (due to \cite{Hajek2016}), the Maximum A Posteriori (MAP) estimator recovers $\strue$ with high probability. Since $\strue$ is chosen uniformly at random from all size-$K$ subsets of $[n]$, then also the Maximum Likelihood estimator (MLE) recovers $\strue$ with high probability. The success of the MLE implies that with high probability,
\begin{equation}
\mathbb{P}(A | \strue) > \max_{S \neq \strue} \mathbb{P}(A | \strue = S).  \label{eq:MLE-inequality}  
\end{equation}
Recall that Algorithm~\ref{alg:PDS} forms the set $U = \{s_1 c_1^{\star} u_1 + s_2 c_2^{\star} u_2: s_1, s_2 \in \{-1,1\}\}$, and chooses the element $u \in U$ which maximizes the likelihood $\mathbb{P}(A | S^{\star} = \hat{S}(u))$. In light of \eqref{eq:MLE-inequality}, it suffices to show that $\strue \in \{\hat{S}(u) : u \in U\}$.

Note that $c_1^{\star}, c_2^{\star}$ are constants due to the scaling of $(\lambda_1^{\star}, u_1^{\star}), (\lambda_2^{\star}, u_2^{\star})$ (cf.~Lemma \ref{lemma:spectral-gap}). Let $s_1, s_2 \in \{-1,1\}$ be such that
\[\left \Vert s_1 c_1^{\star} u_1 + s_ 2 c_2^{\star} u_2 - A\left(\frac{c_1^{\star} }{\lambda_1^{\star}} u_1^{\star} + \frac{c_2^{\star} }{\lambda_2^{\star}} u_2^{\star} \right) \right \Vert_{\infty} \leq \frac{C}{\log \log (n) \sqrt{n}}\]
This is possible by Lemma \ref{lemma:entrywise}, with probability $1-O(n^{-3})$, where $C$ depends on $a$, $b$, and $\rho$. Multiplying through by $\log(n)$, we also obtain
\[\left \Vert \log(n)\left(s_1 c_1^{\star} u_1 + s_2 c_2^{\star} u_2\right) - A \vecS \right \Vert_{\infty} \leq \frac{C \log(n)}{\log \log (n) \sqrt{n}}\]

We first consider the case $a > b > 0$. As shown in the proof of \cite[Theorem 3]{Hajek2016} along with \cite[Lemma 2]{Hajek2016},  we have that for all $C' > 0$,
\begin{align*}
\sqrt{K}\min_{i \in \true} \left(A \vecS\right)_i &\geq \rho \left(\frac{a -b}{\log(a) - \log(b)} \right) \log n +  \frac{C' \log n}{\log \log (n)}\\
\sqrt{K} \max_{i \in [n]\setminus \true} \left(A \vecS\right)_i &\leq \rho \left(\frac{a -b}{\log(a) - \log(b)} \right) \log n
\end{align*}
with probability at least $1 - n^{1 - \rho f(a,b) + o(1)} = 1 - n^{-\Omega(1)}$.
Therefore, we have
\begin{align*}
\log(n) \sqrt{K}\min_{i \in \true} \left(s_1 c_1^{\star} u_1 + s_ 2c_2^{\star} u_2 \right)_i &\geq \rho \left(\frac{a -b}{\log(a) - \log(b)} \right) \log n + \frac{C' \log n}{\log \log (n)} - \frac{C \log(n) \sqrt{K}}{\log \log (n) \sqrt{n}}\\
\log(n) \sqrt{K}\max_{i \in [n] \setminus \true} \left(s_1 c_1^{\star} u_1 + s_2 c_2^{\star} u_2 \right)_i &\leq \rho \left(\frac{a -b}{\log(a) - \log(b)} \right)\log n + \frac{C \log(n) \sqrt{K}}{\log \log (n) \sqrt{n}}
\end{align*}
with probability $1 - n^{-\Omega(1)}$.
Thus, we obtain
\[\min_{i \in \true} (s_1 c_1^{\star} u_1 + s_2 c_2^{\star} u_2)_i > \max_{i \in [n] \setminus \true} (s_1 c_1^{\star} u_1 + s_2 c_2^{\star} u_2)_i\]
with probability $1-n^{-\Omega(1)}$. Therefore, with probability $1-n^{-\Omega(1)}$, the $K$ largest entries of $s_1 c_1^{\star} u_1 + s_2 c_2^{\star} u_2$ correspond to the planted dense subgraph~$\true$.

If $a > b = 0$, then $\max_{i \in [n] \setminus{\true}} (A v_{\true})_i = 0$. Again applying \cite[Lemma 2]{Hajek2016},  we have that for all $C' > 0$,
\begin{align*}
\sqrt{K}\min_{i \in \true} \left(A \vecS\right)_i &\geq \frac{C' \log n}{\log \log (n)}
\end{align*}
with probability at least $1 - n^{1 - \rho a + o(1)} = 1 - n^{1 - \rho f(a,0) + o(1)} =  1 - n^{-\Omega(1)}$.

If $a < b$, then the analysis of the entries of $Av_{\true}$ again appeals to \cite[Lemma 2]{Hajek2016}. We then obtain
\[\max_{i \in \true} (s_1 c_1^{\star} u_1 + s_2 c_2^{\star} u_2)_i < \min_{i \in [n] \setminus \true} (s_1 c_1^{\star} u_1 + s_2 c_2^{\star} u_2)_i\]
with probability $1-n^{-\Omega(1)}$. Therefore,  with probability $1-n^{-\Omega(1)}$, the $K$ smallest entries of $s_1 c_1^{\star} u_1 + s_2 c_2^{\star} u_2$ correspond to the planted dense subgraph~$\true$.
\end{proof}

To prove Lemma \ref{lemma:entrywise}, we need the following two results.
\begin{lemma}\label{lemma:spectral-gap} Let $A$  be the adjacency matrix of a $\PDS$ with parameters satisfying \eqref{eq:PDS-parameters}, and let $A^\star = \E[A]$. 
The eigenvalues of $A^{\star}$ satisfy $\lambda_1^{\star} = (1+o(1))\theta_1 \log(n)$ and $\lambda_2^{\star} = (1+o(1))\theta_2 \log(n)$ for some constants $\theta_1, \theta_2$ depending on $a$, $b$, and $\rho$. 

\end{lemma}

The following result is a special case of \cite[Theorem 5]{Hajek2016}.
\begin{lemma}\label{lemma:spectral-norm-concentration}
Let $A$  be the adjacency matrix of a $\PDS$ with parameters satisfying \eqref{eq:PDS-parameters}. There exists $c_1(a,b,\rho)$ such that
\[\mathbb{P}\left(\Vert A - A^{\star}\Vert_2 \geq c_1 \sqrt{\log (n)} \right) \leq n^{-3}. \]
\end{lemma}
\begin{proof}[Proof of Lemma \ref{lemma:entrywise}]
The result will follow by establishing
\begin{equation}
\min_{s \in \{-1,1\}} \left \Vert  s u_1 -  \frac{A u_1^{\star}}{\lambda_1^{\star}}  \right \Vert_{\infty} \leq \frac{C_1(a,b,\rho)}{\log \log (n) \sqrt{n}}, \label{eq:entrywise-planted-dense-1} \end{equation}
and
\begin{equation}
\min_{s \in \{-1,1\}}  \left \Vert  su_2 -  \frac{A u_2^{\star}}{\lambda_2^{\star}}  \right \Vert_{\infty} \leq \frac{C_2(a,b,\rho)}{\log \log (n) \sqrt{n}}. \label{eq:entrywise-planted-dense-2} 
\end{equation}
In order to establish \eqref{eq:entrywise-planted-dense-1} and \eqref{eq:entrywise-planted-dense-2}, it suffices to verify Assumptions 1-4 of \cite[Theorem 2.1]{Abbe2020}. Letting $\tau = \max\{a,b\}$, we have
$$\Vert A^{\star}\Vert_{2 \to \infty} = \sqrt{\rho n \left(\frac{a \log n}{n}\right)^2 + (1-\rho)n \left(\frac{b \log n}{n}\right)^2} \leq \frac{\tau \log(n)}{\sqrt{n}}.$$ By Lemma \ref{lemma:spectral-gap}, we have that
\[\Delta^{\star} \triangleq  \min\left\{\left( \lambda_1^{\star} - \lambda_2^{\star}\right), |\lambda_1^{\star}|, |\lambda_2^{\star}| \right\} = c \log(n),\]
for some $c = c(a,b,\rho)$. Let 
\[\gamma = \frac{c_1 \sqrt{\log(n)}}{\Delta^{\star}}, \]
where $c_1 = c_1(a,b, \rho)$ is the value from Lemma \ref{lemma:spectral-norm-concentration}. Then $\gamma \Delta^{\star} = c_1 \sqrt{\log(n)}$, which dominates $\frac{\tau \log(n)}{\sqrt{n}}$ for $n$ large enough. Therefore, Assumption 1 of \cite[Theorem 2.1]{Abbe2020} holds.

The second assumption holds since the edges are independent, conditioned on the memberships of the vertices.

By Lemma \ref{lemma:spectral-norm-concentration},
\[\mathbb{P}\left(\Vert A - A^{\star}\Vert_2 \leq \gamma \Delta^{\star} \right) \geq 1 - n^{-3}.\]
Additionally, 
\[\kappa \triangleq \frac{\max\left\{|\lambda_1^{\star}|, |\lambda_2^{\star}| \right\}}{\Delta^{\star}} \]
is a constant by Lemma \ref{lemma:spectral-gap}, and $\gamma \to 0$ as $n \to \infty$. Therefore, the third assumption is verified.

To verify the fourth assumption, choose
\[\varphi(x) = \frac{(2\tau + 4) \log(n)}{\Delta^{\star} \max\left\{1, \log \left(\frac{1}{x}\right)\right\}} \]
Applying \cite[Lemma 7]{Abbe2020}, with $p = \frac{\tau \log n}{n}$ and $\alpha = \frac{4}{\tau}$, we obtain that for a fixed vector $w \in \mathbb{R}^n$ and $m \in [n]$,
\begin{align*}
\mathbb{P}\left(\left|(A - A^{\star})_{m, \cdot} w \right| \leq \frac{(2\tau+4) \log(n)}{\max\left\{1, \log\left(\frac{\sqrt{n}\Vert w \Vert_{\infty}}{\Vert w \Vert_2} \right) \right\}} \Vert w \Vert_{\infty} \right) &\geq 1 - 2n^{-4}\\
\mathbb{P}\left(\left|(A - A^{\star})_{m, \cdot} w \right| \leq  \Delta^{\star} \Vert w \Vert_{\infty} \varphi\left(\frac{\Vert w \Vert_2}{\sqrt{n} \Vert w \Vert_{\infty}} \right) \right) &\geq 1 - 2n^{-4},
\end{align*}
so that Assumption 4 is verified with $\delta_1 = 2n^{-3}$.
\end{proof}

\begin{proof}[Proof of Lemma \ref{lemma:spectral-gap}]
Fix $\epsilon > 0$, and find $\rho' \in [0,1]$ such that $\rho'$ is rational and $|\rho - \rho'| < \epsilon$. Let $s,t$ be the integers such that $\rho' = \frac{s}{t}$ is the minimal representation of $\rho'$ (i.e., $s$ and $t$ have no common factors). 
Let $B$  be the adjacency matrix of a $\PDS$ with parameters satisfying \eqref{eq:PDS-parameters} but with $\rho$ replaced by $\rho'$, and let $B^\star = \E[B]$.
Then $B^{\star}$ has two eigenvectors that correspond to non-zero eigenvalues. The first eigenvector $v_1^{\star}$ of $B^{\star}$ has entries that take on two possible values, one on the first $s$ entries, and one on the remaining $t$ entries. Suppose $\beta_1$ and $\beta_2$ are the values, and note that these are both nonzero. 

Let $u_1^{\star}$ be the first eigenvector of $A^{\star}$, with distinct entries $\alpha_1$ and $\alpha_2$. If $\rho$ were itself rational, then we could let 
$\rho' = \rho$. If additionally $n$ were such that $\rho n = \lfloor \rho n\rfloor$, then $A^{\star}$ would be a scaled version of $B^{\star}$. In that case, we would obtain $\frac{\alpha_2}{\alpha_1} = \frac{\beta_2}{\beta_1}$. For general $\rho$ and fixed $n$, we note that both $\alpha_1$ and $\alpha_2$ are continuous functions of $\frac{K}{n}$, and therefore
\[\left|\frac{\alpha_2}{\alpha_1} - \frac{\beta_2}{\beta_1} \right| \leq f_n(\epsilon),\]
where $\lim_{\epsilon \to 0} f_n(\epsilon) = 0$.

Using $A^{\star} u_1^{\star} = \lambda_1^{\star} u_1^{\star}$, we obtain that for fixed $n$,
\begin{gather*}
\lambda_1^{\star} \alpha_1 = \rho n \cdot \frac{a \log n}{n} \cdot \alpha_1 + (1-\rho) n \cdot \frac{b \log n}{n} \cdot \alpha_2, \quad 
\lambda_1^{\star} = \log(n) \left(\rho a + (1-\rho) b \frac{\alpha_2}{\alpha_1} \right)\\
\log(n) \left(\rho a + (1-\rho) b \left(\frac{\beta_2}{\beta_1}-f_n(\epsilon)\right)  \right) \leq \lambda_1^{\star}  \leq \log(n) \left(\rho a + (1-\rho) b \left(\frac{\beta_2}{\beta_1}+f_n(\epsilon)\right) \right).
\end{gather*}
Choosing a sequence $\{\epsilon_n\}_{n=1}^{\infty}$ such that $\epsilon_n \to 0$, we obtain $\lambda_1^{\star} = (1+o(1)) \theta_1(a, b, \rho) \log(n)$. A similar argument applies to $\lambda_2^{\star}$.
\end{proof}

\section{Analyzing spectral algorithms for censored \PDS}\label{sec:CPDS}
We first provide the proof of Lemma \ref{lemma:connected-conditions}.
\begin{proof}[Proof of Lemma \ref{lemma:connected-conditions}]
There exists $N_0$ such that for all $n > N_0$,
\[t_n p_n \geq \frac{a}{2},~~~~ t_n q_n \geq \frac{b}{2},~~~~ \text{and } |t_n p_n - t_n q_n| \geq \frac{1}{2} | a - b|.\]
We may therefore assume without loss of generality that $p_n, q_n > 0$ and furthermore that $ p_n \neq q_n$ when $a \neq b$, since these statements are true for $n$ sufficiently large. Under the assumptions of the statement,
\begin{align*}
\Delta_+(p,q) &= \max_{x \in [0,1]} xp + (1-x)q - p^xq^{1-x} + x(1-p) + (1-x)(1-q) - (1-p)^x(1-q)^{1-x}\\
&\asymp \max_{x \in [0,1]} xp + (1-x)q - p^xq^{1-x}\\
&= \max_{x \in [0,1]} q + (p-q)x - q \left(\frac{p}{q}\right)^x,
\end{align*}
where the second step uses $p,q = o(1)$. 

We first consider the case $a = b >0$. Then
\begin{align*}
\lim_{n \to \infty} t_n \Delta_+(p,q) &= \lim_{n \to \infty} \max_{x \in [0,1]} q_n t_n + (p_n t_n- q_n t_n) x - q_n t_n \left(\frac{p_n t_n}{q_n t_n} \right)^x = 0 = f(a,a).
\end{align*}

Next, suppose $a \neq b$. Letting $g(x) =  q + (p-q)x - q \left(\frac{p}{q}\right)^x$, we obtain
\begin{align*}
g'(x) &= p -q - q\left(\frac{p}{q}\right)^x \log \left(\frac{p}{q} \right).
\end{align*}
Setting $g'(x^{\star}) = 0$, we obtain the relations
\begin{align*}
\left(\frac{p}{q} \right)^{x^{\star}} &= \frac{p - q}{q \log\left(\frac{p}{q}\right)}\\
x^{\star} &= \frac{1}{\log \left(\frac{p}{q} \right)} \log \left(\frac{p - q}{q \log\left(\frac{p}{q}\right)} \right).
\end{align*}
Substituting, we obtain 
\begin{align} \label{delta-plus-simplification}
\Delta_+(p,q) &\asymp q + \frac{(p-q)}{\log \left(\frac{p}{q} \right)} \log \left(\frac{p - q}{q \log\left(\frac{p}{q}\right)} \right) - \frac{p-q}{\log \left(\frac{p}{q} \right)}.
\end{align}
Multiplying through by $t = t_n$ and using the fact that $\lim_{n \to \infty} p_n t_n = a, \lim_{n\to \infty} q_n t_n = b$, we obtain 
\begin{align*}
t\Delta_+(p,q) &\asymp b + \frac{a-b}{\log \left(\frac{a}{b} \right)} \log \left(\frac{a - b}{b \log\left(\frac{a}{b}\right)} \right) - \frac{a-b}{\log \left(\frac{a}{b} \right)}\\
&= b + \frac{a -b}{\log a - \log b} \left(\log \left(\frac{a - b}{b \log\left(\frac{a}{b}\right)} \right) -1\right)\\
&= b + \frac{a -b}{\log a - \log b} \log \left(\frac{a - b}{eb (\log a - \log b)} \right) \\
&= b - \frac{a -b}{\log a - \log b} \log \left(\frac{eb (\log a - \log b)}{a - b} \right) \\
&= f(b,a)\\
&= f(a,b),
\end{align*}
since $f(\cdot, \cdot)$ is symmetric in its arguments. Since $f(a,b)$ is a constant, we conclude that $\lim_{n \to \infty} t_n \Delta_+(p_n, q_n) = f(a,b)$.
\end{proof}

In the remainder of this section, we analyze Algorithm~\ref{alg:CPDS} and complete the proof of Theorem~\ref{theorem:exact-recovery-censored}. 
The idea again is to apply the method of entrywise perturbation analysis of eigenvectors from~\cite{Abbe2020}. 
The following results will be used to prove Theorem \ref{theorem:exact-recovery-censored}, Part~(1).
\begin{lemma}\label{lemma:MLE}
Suppose the conditions identical to Theorem~\ref{theorem:exact-recovery-censored} hold. Then the \MLE~recovers the planted dense subgraph, with probability $1-o(1)$.
\end{lemma}
\begin{lemma}\label{lemma:entrywise-censored}
Let $c_1, c_2$ be constants. Then with probability $1- O(n^{-3})$,
\begin{equation}
\min_{s_1, s_2 \in \{-1,1\}} \left \Vert  s_1c_1 u_1 + s_2c_2 u_2 - \left(c_1 \frac{A u_1^{\star}}{\lambda_1^{\star}} + c_2 \frac{A u_2^{\star}}{\lambda_2^{\star}} \right)\right \Vert_{\infty} \leq \frac{C}{\log \log (n) \sqrt{n}}, \label{eq:entrywise-planted-dense-censored} \end{equation}
where $C = C(p,q,t, \rho, c_1, c_2)$. 
\end{lemma}

\begin{lemma}\label{lemma:concentration}
Fix $0 < q< p < 1$ and let $y = y(p,q)$ be as in \eqref{defn:y}.
Let $\{X_i\}_{i=1}^m$ be i.i.d.~random variables, where $X_i = 1$ with probability $\alpha p$, $X_i = -y$ with probability $\alpha (1-p)$, and $X_i = 0$ with probability $1-\alpha$. Then for all $\epsilon > 0$, 
\begin{align}\label{eq:conc-1}
\log \left(\mathbb{P}\left(\sum_{i=1}^m X_i \leq \epsilon \log n \right) \right) &\leq - \alpha m \Delta_+(p,q)+ \epsilon \log \left(\frac{p}{q} \right) \log (n) .
\end{align}
Similarly, let $\{Y_i\}_{i=1}^m$ be i.i.d. random variables, where $Y_i = 1$ with probability $\alpha q$, $Y_i = -y$ with probability $\alpha (1-q)$, and $Y_i = 0$ with probability $1-\alpha$. Then
\begin{align}\label{eq:conc-2}
\log \left(\mathbb{P}\left(\sum_{i=1}^m Y_i \geq -\epsilon \log n \right) \right) &\leq   - \alpha m \Delta_+(p,q)+\epsilon \log \left( \frac{p}{q}\right)  \log (n).
\end{align}
\end{lemma}
We now prove the first statement of Theorem \ref{theorem:exact-recovery-censored}.
\begin{proof}[Proof of Theorem \ref{theorem:exact-recovery-censored}]
We first give the proof in the case $p > q$. 
Lemma \ref{lemma:MLE} shows that \eqref{eq:MLE-inequality} holds in this parameter regime. Therefore, as in the proof of Theorem \ref{theorem:dense-subgraph}, it suffices to show that there exist $s_1, s_2 \in \{-1,1\}$ such that the vector $s_1 c_1^{\star} u_1 + s_2 c_2^{\star} u_2$ recovers the planted dense subgraph. 

Let $c_1^{\star}, c_2^{\star}$ be such that $\frac{c_1^{\star} \log(n) }{\lambda_1^{\star}} u_1^{\star} + \frac{c_2^{\star} \log(n)}{\lambda_2^{\star}} u_2^{\star} = \vecS$ as in Algorithm~\ref{alg:weights}. Choose $s_1, s_2 \in \{-1,1\}$ such that 
\[\left \Vert s_1 c_1^{\star} u_1 + s_2 c_2^{\star} u_2 - A\left(\frac{c_1^{\star} }{\lambda_1^{\star}} u_1^{\star} + \frac{c_2^{\star} }{\lambda_2^{\star}} u_2^{\star} \right) \right \Vert_{\infty} \leq \frac{C}{\log \log (n) \sqrt{n}}.\]
This is possible by Lemma \ref{lemma:entrywise-censored} with probability $1-O(n^{-3})$, where $C$ depends on $p$, $q$, $t$, and $\rho$. Multiplying through by $\log(n)$, we also obtain
\[\left \Vert \log(n)\left(s_1 c_1^{\star} u_1 + s_2 c_2^{\star} u_2\right) - A \vecS \right \Vert_{\infty} \leq \frac{C \log(n)}{\log \log (n) \sqrt{n}}\]
We now apply Lemma \ref{lemma:concentration} with $m = \rho n$. In particular, if $i\in \true$, then \eqref{eq:conc-1} applies and if $i\notin \true$ then \eqref{eq:conc-2} applies. 
Then 
\[\alpha m \Delta_+ (p,q) =   t\rho \Delta_+(p,q)\log(n).\]
Since $t\rho \Delta_+(p,q) > 1$, there exists $\epsilon > 0$ such that
\begin{align*}
&\sqrt{K}\min_{i \in \true} (A \vecS)_i \geq \epsilon \log(n), \quad \sqrt{K}\max_{i \notin \true} (A \vecS)_i \leq -\epsilon \log(n).
\end{align*}
with probability $1-n^{-\Omega(1)}$. Therefore, 
\begin{align*}
\log(n) \sqrt{K} \min_{i \in \true} (s_1 c_1^{\star} u_1 + s_2 c_2^{\star} u_2)_i \geq \epsilon \log(n) - \frac{C \log(n) \sqrt{K}}{\log \log(n) \sqrt{n}}\\
\log(n) \sqrt{K} \max_{i \in [n] \setminus \true} (s_1 c_1^{\star} u_1 + s_2 c_2^{\star} u_2)_i \leq -\epsilon \log(n) + \frac{C \log(n) \sqrt{K}}{\log \log(n) \sqrt{n}}.
\end{align*}
with probability $1-n^{-\Omega(1)}$.
We conclude that with probability $1-n^{-\Omega(1)}$,
\[ \min_{i \in \true} (s_1 c_1^{\star} u_1 + s_2 c_2^{\star} u_2)_i >  \max_{i \in [n] \setminus \true} (s_1 c_1^{\star} u_1 + s_2 c_2^{\star} u_2)_i. \]

Next, consider the case $p < q$ (the equality case is ruled out since $\Delta_+(p,p) = 1$). The only change is in the application of Lemma \ref{lemma:concentration}; if $i \in S^{\star}$, then \eqref{eq:conc-2} applies and if $i \not \in \true$ then \eqref{eq:conc-1} applies. We then conclude that with probability $1-n^{-\Omega(1)}$,
\[ \max_{i \in \true} (s_1 c_1^{\star} u_1 + s_2 c_2^{\star} u_2)_i <  \min_{i \in [n] \setminus \true} (s_1 c_1^{\star} u_1 + s_2 c_2^{\star} u_2)_i. \]

\end{proof}

\begin{proof}[Proof of Lemma \ref{lemma:MLE}]
As noted in the proof of Theorem \ref{theorem:exact-recovery-censored} Part (1), it suffices to consider the case $p > q$. 
We first characterize the \MLE. Given a signed  adjacency matrix $A$ and a set $S \subseteq [n]$, let 
\begin{align*}
    E^+_1(A, S) = \left|\{\{i,j\} : A_{ij} = 1, i, j \in S \}\right|, \quad E^-_1(A, S) = \left|\{\{i,j\} : A_{ij} = 1, \{i \not \in S \text{ or } j \not \in S\} \}\right| \\
    E^+_y(A, S) = \left|\{\{i,j\} : A_{ij} = -y, i, j \in S \}\right|, \quad E^-_y(A, S) = \left|\{\{i,j\} : A_{ij} = -y, \{i \not \in S \text{ or } j \not \in S\} \}\right|.
\end{align*}
Given this notation, we can write the log-likelihood as
\begin{align*}
&\log \mathbb{P}(A|\strue) \\
&= \log(\alpha p) E_1^+(A, \strue) + \log(\alpha(1-p)) E_y^+(A,\strue) + \log(\alpha q) E_1^-(A, \strue) + \log(\alpha(1-q)) E_y^-(A,\strue) \\
&~~~~~~+ \log(1-\alpha) \left(\binom{n}{2} -E_1^+(A, \strue) - E_y^+(A, \strue)-E_1^-(A, \strue)-E_y^-(A, \strue)  \right).    
\end{align*}
When taking the difference of log likelihoods below, the final term does not contribute. We obtain
\begin{align*}
&\log\left(\mathbb{P}(A | \strue\right))
- \log\left(\mathbb{P}(A | \strue = S \right))\\ &= \log(\alpha p) \left(E_1^+(A,\strue) - E_1^+(A, S) \right) + \log(\alpha(1-p))\left(E_y^+(A,\strue) - E_y^+(A, S) \right)\\
&~~~+\log(\alpha q) \left(E_1^-(A,\strue) - E_1^-(A, S) \right) + \log(\alpha(1-q))\left(E_y^-(A,\strue) - E_y^-(A, S) \right)\\
&= \log(p) \left(E_1^+(A,\strue) - E_1^+(A, S) \right) + \log(1-p)\left(E_y^+(A,\strue) - E_y^+(A, S) \right)\\
&~~~+\log(q) \left(E_1^-(A,\strue) - E_1^-(A, S) \right) + \log(1-q)\left(E_y^-(A,\strue) - E_y^-(A, S) \right)\\
&= \log \left( \frac{p}{q}\right) \left(E_1^+(A,\strue) - E_1^+(A, S) \right) + \log \left( \frac{1-p}{1-q}\right) \left(E_y^+(A,\strue) - E_y^+(A, S) \right).
\end{align*}
It follows that the \MLE~maximizes $\log ( \frac{p}{q}) E_1^+(A,S) + \log ( \frac{1-p}{1-q}) E_y^+(A,S)$ over $S \subseteq [n]$, where $|S| = \lfloor \rho n \rfloor$. Let $F = \{S \subseteq [n] : |S| = \lfloor \rho n \rfloor\}$ denote the set of feasible sets. By rearranging, we find that the MLE computes
\[\argmax_{S \in F} \left\{E_1^+(A,S) - y E_y^+(A,S)\right\}.\]
Since $t\rho \Delta_+(p,q) > 1$, it is possible to choose $\epsilon$ satisfying
\[0 < \epsilon < \frac{t\rho\Delta_+(p,q) -1}{\log\left(\frac{p}{q}\right)}.\]
First we will show that with high probability, the \MLE~will not select any $S \in F$ such that $0 < |S \setminus \strue| \leq \frac{\epsilon}{t(1+y)} n$. 
To this end, we will leverage the concentration of partial row sums of $A$. Let $X_i = \sum_{j \in \strue} A_{ij}$. Lemma \ref{lemma:concentration} along with a union bound implies that 
\begin{align}
\mathbb{P}\left(\left\{\min_{i \in \strue} X_i \leq \epsilon \log n\right\} \cup \left\{\max_{i \in \strue} X_i \geq -\epsilon \log n\right\} \right) &\leq n^{1 + \epsilon \log\left(\frac{p}{q}\right) - t\rho \Delta_+(p,q)} = n^{-\Theta(1)}. \label{eq:simultaneous-concentration}
\end{align}
Comparing $\strue$ with $S \in F$,
\begin{align}
&E_1^+(A,S) - y E_y^+(A,S) - \left(E_1^+(A,\strue) - y E_y^+(A,\strue)\right)\nonumber \\
&= \sum_{i \in  S \setminus \strue, j \in \strue \cap  S} A_{ij} + \sum_{i,j \in S \setminus \strue: i < j}A_{ij} -  \sum_{i \in \strue \setminus S, j \in S \cap \strue} A_{ij}  - \sum_{i,j \in \strue \setminus S: i < j}A_{ij}. \label{eq:objective-value-difference}
\end{align}
We can rewrite the first term of \eqref{eq:objective-value-difference} as follows:
\begin{equation}
\sum_{i \in  S \setminus \strue, j \in \strue \cap  S} A_{ij} = \sum_{i \in  S \setminus \strue, j \in \strue} A_{ij} - \sum_{i \in  S \setminus \strue, j \in \strue \setminus S} A_{ij}. \label{eq:objective-value-difference-1}
\end{equation}
Under the high probability event in \eqref{eq:simultaneous-concentration}, we have $\sum_{i \in  S \setminus \strue, j \in \strue} A_{ij} < - \epsilon \log n|S\setminus S^\star|$. In order to bound the second term of \ref{eq:objective-value-difference-1}, we bound the number of non-zero entries in the summation. Let $H$ be the graph on $n$ vertices which contains an edge between $i$ and $j$ whenever the edge status is revealed. 
Then $H$ is distributed as an Erd\H{o}s-R\'enyi random graph with parameters $n$ and $\nicefrac{t\log n}{n}$.
For two sets $U, W \subseteq [n]$, let $r(U,W)$ denote the number of edges between $U$ and $W$ in $H$ (where we double-count any edge connecting two vertices in the intersection $U \cap W$). By \cite[Corollary 2.3]{Krivelevich2006},
\begin{equation}
\left|r(U, W) - \frac{t \log n}{n} |U| \cdot |W| \right| = O\left(\sqrt{|U| \cdot |W| t \log (n)} \right), \label{eq:pseudorandom}
\end{equation}
for all $U, W$ with high probability. Therefore, the number of revealed edge statuses between $S \setminus \strue$ and $\strue \setminus S$ is at most
\[\frac{t \log n}{n}\left(|S \setminus \strue|\right)^2 + O\left(|S \setminus \strue| \sqrt{\log n} \right).\]
In the worst case, all of the revealed statuses are nonedges, contributing $-y$. Combining these observations,
\[ \sum_{i \in  S \setminus \strue, j \in \strue \setminus S} A_{ij} \geq -y\left(\frac{t \log n}{n}\left(|S \setminus \strue|\right)^2 + O\left(|S \setminus \strue| \sqrt{\log n} \right) \right).\]
under the high probability event \eqref{eq:pseudorandom}. Therefore, we can bound \eqref{eq:objective-value-difference-1} by
\[\sum_{i \in  S \setminus \strue, j \in \strue \cap  S} A_{ij} \leq -\epsilon \log n|S\setminus S^\star| + y\left(\frac{t \log n}{n}\left(|S \setminus \strue|\right)^2 + O\left(|S \setminus \strue| \sqrt{\log n} \right) \right),\]
under the high probability events \eqref{eq:simultaneous-concentration} and \eqref{eq:pseudorandom}. We can similarly analyze the third term in \eqref{eq:objective-value-difference}:
\begin{align*}
 -  \sum_{i \in \strue \setminus S, j \in S \cap \strue} A_{ij} &=  -\sum_{i \in \strue \setminus S, j \in \strue} A_{ij} +  \sum_{i \in \strue \setminus S, j \in \strue \setminus S} A_{ij}  \\
 &\leq -\epsilon \log n + \frac{t \log n}{n}\left(|S \setminus \strue|\right)^2 + O\left(|S \setminus \strue| \sqrt{\log n} \right).
\end{align*}
Finally, we bound the second and fourth terms in \eqref{eq:objective-value-difference}:
\begin{align*}
\sum_{i,j \in S \setminus \strue: i < j}A_{ij} - \sum_{i,j \in \strue \setminus S : i< j}A_{ij} &\leq \frac{1}{2} (1+y) \left(\frac{t \log n}{n}\left(|S \setminus \strue|\right)^2 + O\left(|S \setminus \strue| \sqrt{\log n} \right) \right).
\end{align*}
Putting everything together, we bound \eqref{eq:objective-value-difference}:
\begin{align*}
&E_1^+(A,S) - y E_y^+(A,S) - \left(E_1^+(A,\strue) - y E_y^+(A,\strue)\right)\\
&\leq -2\epsilon |S \setminus \strue| \log(n) + \frac{3}{2}(1+y)\left(\frac{t \log n}{n}\left(|S \setminus \strue|\right)^2 + O\left(|S \setminus \strue| \sqrt{\log n} \right) \right) \\
&= |S \setminus \strue| \log n\left[-2\epsilon + \frac{3}{2}(1+y) \frac{t}{n} |S \setminus \strue| + O \left(\frac{1}{\sqrt{\log n}} \right) \right].
\end{align*}
If $|S \setminus \strue| \leq \frac{\epsilon}{t(1+y)} n$, then the difference in objective values will be negative for large enough $n$. We conclude that with high probability, the MLE will not select such an $S$.

It remains to show that with high probability, the MLE will not select any $S \in F$ such that $|S \setminus \strue| > \frac{\epsilon}{t(1+y)} n.$ Fix $\frac{\epsilon}{t(1+y)} n < m \leq n$, and let $S \in F$ be such that $|S \setminus \strue| = m$. Let $N(m) = \binom{m}{2} + m(K-m)$. 
Let $\{X_i\}_{i=1}^{N(m)}$ be i.i.d.~random variables, where $X_i = 1$ with probability $\alpha p$, $X_i = -y$ with probability $\alpha (1-p)$, and $X_i = 0$ with probability $1-\alpha$. Also let $\{Y_i\}_{i=1}^{N(m)}$ be i.i.d.~random variables independent of the $X_i$'s, where $Y_i = 1$ with probability $\alpha q$, $Y_i = -y$ with probability $\alpha (1-p)$, and $Y_i = 0$ with probability $1-\alpha$. Examining \eqref{eq:objective-value-difference}, we see that the difference $E_1^+(A,S) - y E_y^+(A,S) - \left(E_1^+(A,\strue) - y E_y^+(A,\strue)\right)$ has the same distribution as $\sum_{i=1}^{N(m)} Y_i - \sum_{i=1}^{N(m)} X_i$. By Lemma \ref{lemma:concentration},
\begin{align*}
\log \left(\mathbb{P}\left(\sum_{i=1}^{N(m)} X_i \leq \delta \log n \right) \right) &\leq   - \alpha N(m) \Delta_+ (p,q) + \delta \log \left(\frac{p}{q} \right) \log (n)
\end{align*}
and
\begin{align*}
\log \left(\mathbb{P}\left(\sum_{i=1}^{N(m)} Y_i \geq -\delta \log n \right) \right) &\leq    - \alpha N(m) \Delta_+(p,q) +\delta \log \left( \frac{p}{q}\right)  \log (n)
\end{align*}
for any $\delta > 0$. Simplifying the bounds using $N(m) = \binom{m}{2} + m (K-m)$, we obtain
\begin{align*}
\delta \log \left( \frac{p}{q}\right)  \log (n)  - \alpha N(m) \Delta_+(p,q) & = - \Theta(n \log n).
\end{align*}
It follows that
\begin{align*}
\mathbb{P}\left(\sum_{i=1}^{N(m)} Y_i - \sum_{i=1}^{N(m)} X_i \geq 0 \right) \leq e^{-\Theta(n \log n)},
\end{align*} 
which also implies that $E_1^+(A,S) - y E_y^+(A,S) - \left(E_1^+(A,\strue) - y E_y^+(A,\strue)\right) < 0$ with probability at least $1 - e^{-\Theta(n \log n)}$.

On the other hand, $|F| \leq 2^n$. 
By a union bound, the probability that the \MLE~selects some $S$ satisfying $|S \setminus \strue| > \frac{\epsilon}{t(1+y)} n$ is at most
\[2^n e^{-\Theta(n \log n)} = o(1).\]
We conclude that with high probability, the \MLE~selects $\strue$.
\end{proof}

\begin{proof}[Proof of Lemma \ref{lemma:entrywise-censored}]
The result will follow by establishing
\begin{equation}
\min_{s \in \{-1,1\}} \left \Vert  s u_1 -  \frac{A u_1^{\star}}{\lambda_1^{\star}}  \right \Vert_{\infty} \leq \frac{C_1(p,q,t,\rho)}{\log \log (n) \sqrt{n}}, \label{eq:entrywise-censored-planted-dense-1} \end{equation}
and
\begin{equation}
\min_{s \in \{-1,1\}}  \left \Vert  su_2 -  \frac{A u_2^{\star}}{\lambda_2^{\star}}  \right \Vert_{\infty} \leq \frac{C_2(p,q,t,\rho)}{\log \log (n) \sqrt{n}}. \label{eq:entrywise-censored-planted-dense-2} \end{equation}
In order to establish \eqref{eq:entrywise-censored-planted-dense-1} and \eqref{eq:entrywise-censored-planted-dense-2}, it suffices to verify Assumptions 1-4 of \cite[Theorem 2.1]{Abbe2020}.

By a similar argument as the proof of Lemma \ref{lemma:spectral-gap}, we have $\Delta^{\star} = c_1 \log(n)$, for some $c_1 = c_1(p,q,\rho, t)$. 
It follows from \cite[Theorem 5]{Hajek2016} that
\[\mathbb{P}\left(\Vert A - A^{\star}\Vert_2 \leq c_2 \sqrt{\log n}\right) \geq 1 - n^{-3}\]
for some $c_2 = c_2(p,q,\rho,t)$. We have
\[\Vert A^{\star}\Vert_{2 \to \infty} = \sqrt{\rho n \left(\alpha p - \alpha(1-p)y \right)^2 + (1-\rho) n \left(\alpha q - \alpha(1-q)y \right)^2} = \Theta\left(\frac{\log n}{\sqrt{n}}\right).\]
Let $\gamma = \frac{c_2 \sqrt{\log(n)}}{\Delta^{\star}}$. Then $\gamma \Delta^{\star}$ dominates $\Vert A^{\star} \Vert_{2 \to \infty}$ for $n$ large enough, verifying Assumption 1.

The second and third assumptions hold by the same reasoning as the proof of Lemma \ref{lemma:entrywise}. The fourth assumption is verified identically as in the proof of \cite[Proposition 5.1]{Dhara2021}.
\end{proof}

\begin{proof}[Proof of Lemma \ref{lemma:concentration}]
Let $\lambda > 0$. Exponentiating and applying the Markov inequality,
\begin{align*}
\mathbb{P}\left(\sum_{i=1}^m X_i \leq \epsilon \log n \right)&= \mathbb{P}\left(- \lambda \sum_{i=1}^m X_i \geq - \lambda \epsilon \log n \right)\\
&= \mathbb{P}\left(\exp\left(- \lambda \sum_{i=1}^m X_i\right) \geq \exp\left(- \lambda \epsilon \log n\right) \right)\\
&\leq n^{\lambda \epsilon} \mathbb{E}\left[\exp\left(- \lambda \sum_{i=1}^m X_i\right) \right]= n^{\lambda \epsilon} \left(\mathbb{E}[e^{-\lambda X_1}]\right)^m.
\end{align*}
Using $\log(1+x) \leq x$ for $x > -1$,
\begin{align*}
\log\left(\mathbb{E}[e^{-\lambda X_1}]\right) &= \log \left(e^{-\lambda} \alpha p + e^{\lambda y} \alpha(1-p) + 1-\alpha \right)\\
&\leq \alpha \left(e^{-\lambda}p + e^{\lambda y} (1-p) - 1 \right).
\end{align*}
Therefore,
\begin{align*}
\log \left(\mathbb{P}\left(\sum_{i=1}^m X_i \leq \epsilon \log n \right) \right) &\leq \lambda \epsilon \log (n) + \alpha m \left(e^{-\lambda}p + e^{\lambda y} (1-p) - 1 \right) 
\end{align*}
Set $\lambda = (1-x^{\star}) \log \left( \frac{p}{q}\right)$, where
$x^{\star}$ is the minimizer of 
\begin{eq}
p^x q^{1-x} + (1-p)^x (1-q)^{1-x}
\end{eq}
over $x \in [0,1]$. 
We claim that $0<x^\star<1$. Indeed, taking $x \in \{0,1\}$ yields a value of $1$, while $x = \frac{1}{2}$ yields $\sqrt{pq} + \sqrt{(1-p)(1-q)}$, which is strictly less than $1$ when $p \neq q$. Using continuity of $p^x  q^{1-x} + (1-p) ^x (1-q)^{1-x}$ with respect to $x$, we conclude that $x^\star\in (0,1)$.
Next, using that
\[e^{-\lambda} p + e^{\lambda y} (1-p) =  p^{x^{\star}} q^{1-x^{\star}} + (1-p)^{x^{\star}} (1-q)^{1-x^{\star}} = -\Delta_+(p,q) + 1,\]
we obtain
\begin{align*}
\log \left(\mathbb{P}\left(\sum_{i=1}^m X_i \leq \epsilon \log n \right) \right) &\leq - \alpha m \Delta_+(p,q)+ (1-x^{\star}) \epsilon \log \left(\frac{p}{q} \right)  \log (n) \\
&\leq  - \alpha m \Delta_+(p,q)+ \epsilon \log \left(\frac{p}{q} \right)  \log (n).
\end{align*}
Similarly,
\begin{align*}
\mathbb{P}\left(\sum_{i=1}^m Y_i \geq -\epsilon \log n \right)&= \mathbb{P}\left(\lambda \sum_{i=1}^m Y_i \geq - \lambda \epsilon \log n \right)\\
&= \mathbb{P}\left(\exp\left(\lambda \sum_{i=1}^m Y_i\right) \geq \exp\left( -\lambda \epsilon \log n\right) \right)\\
&\leq n^{\lambda \epsilon} \mathbb{E}\left[\exp\left(\lambda \sum_{i=1}^m Y_i\right) \right]= n^{\lambda \epsilon} \left(\mathbb{E}[e^{\lambda Y_1}]\right)^m.
\end{align*}and
\begin{align*}
\log\left(\mathbb{E}[e^{\lambda Y_1}]\right) &= \log \left(e^{\lambda} \alpha q + e^{-\lambda y} \alpha(1-q) + 1-\alpha \right)\\
&\leq \alpha \left(e^{\lambda}q + e^{-\lambda y} (1-q) - 1 \right).
\end{align*}
Set $\lambda = x^{\star} \log \left( \frac{p}{q}\right)$. Then
\[e^{\lambda} q + e^{-\lambda y} (1-q) =  p^{x^{\star}} q^{1-x^{\star}} + (1-p)^{x^{\star}} (1-q)^{1-x^{\star}} = -\Delta_+(p,q) + 1,\]
and we obtain
\begin{align*}
\log \left(\mathbb{P}\left(\sum_{i=1}^m Y_i \geq -\epsilon \log n \right) \right) &\leq  - \alpha m \Delta_+(p,q)+ x^{\star} \epsilon \log \left( \frac{p}{q}\right)  \log (n) \\
&\leq  - \alpha m \Delta_+(p,q)+ \epsilon \log \left( \frac{p}{q}\right)  \log (n) .
\end{align*}
The proof of Lemma~\ref{lemma:concentration} is now complete. 
\end{proof}

\section{Impossibility of recovery for the \CPDS~model} \label{sec:impossible}
\subsection{Proof of Theorem \ref{theorem:impossibility2}}
We start by identifying a sufficient condition for impossibility of recovering the planted dense subgraph. 
Let $\cS$ be the space of possible assignments of vertices to the dense subgraph, i.e., all possible subsets $S$ of size~$K$. 
Recall that $\true \in \cS$ is the true underlying dense subgraph, which is a uniform subset of size $K$. 
We write $g$ as a generic notation to denote the observed value of the edge-labeled graph $G$ consisting of present, absent and censored edges. 
Also, let $\cG$ be the space of all possible values of $G$. 

The optimal estimator of $\true$ is the Maximum A Posteriori (\textsc{MAP}) estimator. Given a realization $g \in \mathcal{G}$, the \textsc{MAP} estimator outputs $\map \in \argmax_{\mathcal{S}} \mathbb{P}(S_0 = S | G = g)$, choosing uniformly at random from the argmax set. Since $\true$ is uniformly distributed on $\mathcal{S}$, we have the following equality between sets:
\[\argmax_{\mathcal{S}} \mathbb{P}(\true = S | G = g) = \argmax_{\mathcal{S}} \mathbb{P}(G = g | \true = S);\]
i.e. the \textsc{MAP} estimator coincides with the Maximum Likelihood estimator. Due to the equivalence, we obtain the following condition for failure of the \textsc{MAP} estimator. 

\begin{proposition}\label{proposition:equiprobable}
Fix $\delta>0$. Suppose that there is  $\cG' \subset \cG $ with $\PR(G \in \cG' | \true)\geq \delta$ such that the following holds for any $g \in \cG'$: There are $k$ pairs of vertices $\{(u_i,v_i): i\in [k]\}$ with $u_i\in \true$ and $v_i\notin \true$ such that if $\tilde{S}^\star := (\true \cup \{v_i\} )\setminus\{u_i\}$, 
then  $\mathbb{P}(G = g \mid \true) = \mathbb{P}(G = g \mid \tilde{S}^\star)$. Then, conditionally on $\true$, 
the MAP estimator $\map$ fails in exact recovery with probability at least $\delta\big(1 - \frac{1}{k}\big)$.
\end{proposition}
\begin{proof}
By our underlying condition, whenever $g\in \cG'$,  the true assignment is such that swapping $u_i$ and $v_i$ from $\true$  results in an equiprobable assignment. In that case, the algorithm is incorrect with probability at least $1- \frac{1}{k}$.
 Therefore, 
 \begin{eq}
 \PR(\map \neq \true ~\big|~ \true) \geq \PR(\map \neq\true ~\big|~ G \in \cG', \true) \PR(G \in \cG' ~\big|~ \true) \geq \delta \Big(1- \frac{1}{k}\Big),
 \end{eq}
 and the proof follows.
\end{proof}
Next we complete the proof of the impossibility result. 
We will use the following Poisson approximation result: 
\begin{lemma}
\label{fact:stirling}
Let $\{W_i\}_{i=1}^{m}$ be i.i.d.~from a distribution taking three values $a,b,c$ and $\mathbb{P}(W_i = a) = \alpha \psi$, $\mathbb{P}(W_i = b) = \alpha (1-\psi)$, and   $\mathbb{P}(W_i = c) = 1- \alpha$. 
Let $N_x:= \#\{i: W_i =x\}$ for $x = a,b,c$. If $\alpha=m^{o(1)-1}$ then the probability distribution of $(N_a,N_b)$ is within a total variation distance of $m^{o(1)-1}$ of the product of a Poisson distribution of mean $\alpha \psi m$ and a Poisson distribution of mean $\alpha(1-\psi)m$.
\end{lemma}

\begin{proof}
For each $i$, let $\overline{W}_i$ be $(1,0)$ if $W_i=a$, $(0,1)$ if $W_i=b$, and $(0,0)$ if $W_i=c$. Also, for each $i$ let $\widehat{W_i}$ be drawn from $P(\alpha \psi)\times P(\alpha(1-\psi))$ where $P(\lambda)$ is the Poisson distribution with mean $\lambda$. The probability distributions of $\overline{W}_i$ and $\widehat{W}_i$ are within a total variation distance of $O(\alpha^2)=m^{o(1)-2}$ of each other. Furthermore, all of the $\overline{W}_i$ and $\widehat{W}_i$ are mutually independent. So, the probability distribution of $(N_a,N_b)$ is within a total variation distance of $O(\alpha^2m)=m^{o(1)-1}$ of the probability distribution of $\sum_i \widehat{W}_i$, which is $P(\alpha \psi m)\times P(\alpha (1-\psi) m)$.
\end{proof}
This allows us to finally prove Theorem \ref{theorem:impossibility2}.
\begin{proof}[Proof of Theorem \ref{theorem:impossibility2}]
We will show that with high probability, there exist $k=\omega(1)$ pairs of vertices $\{(u_i,v_i):i\in [k], u_i \in \true, v_i \notin \true\}$ such that swapping their labels results in an equiprobable graph instance. By Proposition~\ref{proposition:equiprobable}, this would show that exact recovery fails with probability $1 - o(1)$ for any algorithm.  First, due to \eqref{impossibility-condition}, there exists a sufficiently small enough constant $\delta>0$ such that $t\rho \Delta_+(p,q)<1-2\delta$ for all sufficiently large values of $n$.

For $j=1,2$, let $V_j$ be sets of $2n^{1-\delta}$ randomly selected vertices from $\true$ and $(\true)^c$, respectively. 
Let $V = V_1 \cup V_2$. 
Next, let $V'$ be the set of all vertices in $V$ whose connections to all other vertices in $V$ are censored. We claim that $|V'|> 3n^{1-\delta}$ with probability $1-o(1)$. To see this, observe that the expected number of revealed connections between vertices in $V$ is at most $\alpha \binom{4n^{1-\delta}}{2}\asymp 8n^{1-2\delta}t\log n$, so with high probability there are fewer than $n^{1-\delta}/2$ such connections. Therefore with high probability there are fewer than $n^{1-\delta}$ vertices in $V$ with at least one neighbor in $V$, from which the claim follows.
Let 
\[x^{\star} = \argmax_{x \in [0,1]} \big(1-p^xq^{1-x} - (1-p)^x(1-q)^{1-x} \big)\]
so that $\Delta_+(p,q) = 1-p^{x^\star}q^{1-x^\star} - (1-p)^{x^\star}(1-q)^{1-x^\star}$. 
For a vertex $v$, we call $(d_1,d_2)$ its \emph{degree profile}, where $d_1$ (respectively $d_2$) is the number of present edges (respectively absent edges) $v$ has to the planted dense subgraph $\true$. 
Consider the degree profile 
\begin{align} \label{bad-degree-profile}
    d&=\big(\lceil p^{x^*}q^{1-x^*}t \rho \log(n)\rceil ,\lceil (1-p)^{x^*}(1-q)^{1-x^*}t \rho\log(n)\rceil\big) =: (d_{1},d_{2}).
\end{align}
We will show that there are vertices in $\true$ and $(\true)^c$ with the degree profile given by \eqref{bad-degree-profile}. 
For $v\in \true $ (respectively $v\notin \true$), let $\overline{p}_1$ (respectively~$\overline{p}_2$) be the probability that $v$ has degree profile $d$, conditioned on~$v\in V'$. 
By Lemma~\ref{fact:stirling}, we have that $\overline{p}_1$ is within $n^{o(1)-1}$ of the probability of drawing $d$ from a multidimensional Poisson distribution with mean $\left(\alpha p(K-|V_1|), \alpha(1-p)(K-|V_1|)\right)$ and $\overline{p}_2$ is within $n^{o(1)-1}$ of the probability of drawing $d$ from a multidimensional Poisson distribution with mean $\left(\alpha q(K-|V_1|), \alpha(1-q)(K-|V_1|)\right)$. So,

\begin{eq} \label{eq:min-prob}
   &\min\{\overline{p}_1, \overline{p}_2\} \\
   &\sim \min \bigg\{\e^{-t\rho \log n}\frac{(t\rho p \log n)^{d_{1}}}{d_{1}!} \frac{(t\rho (1-p) \log n)^{d_{2}}}{d_{2}!} , \e^{-t\rho \log n}\frac{(t\rho q \log n)^{d_{1}}}{d_{1}!} \frac{(t\rho (1-q) \log n)^{d_{2}}}{d_{2}!} \bigg\}\pm n^{o(1)-1}\\
   & \sim \e^{-t\rho \log n} \min \bigg\{ \frac{(t\rho p \log n)^{d_{1}}(t\rho (1-p) \log n)^{d_{2}}}{2\pi \e^{-d_1-d_2} d_1^{d_1+1/2}d_2^{d_2+1/2}},\frac{(t\rho q \log n)^{d_{1}}(t\rho (1-q) \log n)^{d_{2}}}{2\pi \e^{-d_1-d_2} d_1^{d_1+1/2}d_2^{d_2+1/2}}\bigg\}\pm n^{o(1)-1}\\
   &= \frac{\e^{-t\rho \log n + d_1+d_2}}{2\pi \sqrt{d_1d_2}} \min \bigg\{ \bigg(\frac{t\rho p\log n}{d_1} \bigg)^{d_1} \times \bigg(\frac{t\rho (1-p)\log n}{d_2} \bigg)^{d_2}, \\
   & \hspace{8cm}\bigg(\frac{t\rho q\log n}{d_1} \bigg)^{d_1} \times \bigg(\frac{t\rho (1-q)\log n}{d_2} \bigg)^{d_2} \bigg\}\pm n^{o(1)-1} \\
   & \sim \frac{\e^{-t\rho \log n + d_1+d_2}}{2\pi \sqrt{d_1d_2}} n^{o(1)}\min \bigg\{ \bigg(\frac{p}{q}\bigg)^{(1-x^*)d_1} \bigg(\frac{1-p}{1-q}\bigg)^{(1-x^*)d_2}, \bigg(\frac{q}{p}\bigg)^{x^*d_1} \bigg(\frac{1-q}{1-p}\bigg)^{x^*d_2}\bigg\}\pm n^{o(1)-1}, 
\end{eq}
and moreover,
\begin{align*}
    \e^{-t\rho \log n + d_1+d_2} \asymp \e^{-t\rho \log n [1 -p^{x^*}q^{1-x^*} - (1-p)^{x^*}(1-q)^{1-x^*}) ]} = n^{-t\rho \times \Delta_+(p,q)}.
\end{align*}
Next, let $f(x) = p^x q^{1-x} + (1-p)^x (1-q)^{1-x}$. Since $f$ attains its minimum at $x^*$, we have $f'(x^*) =0$ and therefore
\begin{align*}
    p^{x^*} q^{1-x^*} \log \frac{p}{q} + (1-p)^{x^*} (1-q)^{1-x^*} \log \frac{1-p}{1-q} = 0. 
\end{align*}
Thus, 
\begin{align*}
    &\bigg(\frac{p}{q}\bigg)^{(1-x^*)d_1} \bigg(\frac{1-p}{1-q}\bigg)^{(1-x^*)d_2} \\
    & \quad = \exp \bigg((1-x^*)t\rho \log n  \bigg[p^{x^*} q^{1-x^*} \log \frac{p}{q} + (1-p)^{x^*} (1-q)^{1-x^*} \log \frac{1-p}{1-q}  \bigg]\pm o(\log(n))\bigg) = n^{o(1)}
\end{align*}
and similarly, 
\begin{align*}
    &\bigg(\frac{q}{p}\bigg)^{x^*d_1} \bigg(\frac{1-q}{1-p}\bigg)^{x^*d_2} \\
    &\quad = \exp \bigg(-x^*t\rho \log n  \bigg[p^{x^*} q^{1-x^*} \log \frac{p}{q} + (1-p)^{x^*} (1-q)^{1-x^*} \log \frac{1-p}{1-q}  \bigg]\pm o(\log(n))\bigg) = n^{o(1)}.
\end{align*}
Thus, \eqref{eq:min-prob} reduces to 
\begin{align*}
    \min\{\overline{p}_1, \overline{p}_2\} \asymp \frac{n^{-\Delta_+-o(1)}}{2\pi \sqrt{d_1d_2}} \pm n^{o(1)-1}\ge n^{2\delta-1+o(1)},
\end{align*}
where the last step uses $t\rho \Delta_+(p,q)\le 1-2\delta$ for large $n$.

For $v \in V'$, let $Y(v)$ be the indicator that $v$ has exactly $d_1$ present edges and $d_2$ absent edges to $S$. Note that the random variables in the set $\{Y(v)\}_{v \in V'}$ are mutually independent conditionally on $V'$.
Finally, observe that if $u \in V' \cap V_1$ and $v \in V' \cap V_2$ satisfy $Y(u) = Y(v) = 1$, then switching the labels of $u$ and $v$ results in an equiprobable outcome.

Let 
\[Y_1 = \sum_{v \in V' \cap V_1} Y(v) ~~~\text{and}~~~ Y_2 = \sum_{v \in V' \cap V_2} Y(v).\]
It suffices to show that there is a function $f(n) = \omega(1)$ such that $Y_1, Y_2 \geq f(n)$ with probability $1-o(1)$. 

Similarly to the proof of \cite[Theorem 2.1]{Dhara2021}, we can show
\begin{align}
\mathbb{P}\left(Y_1 \leq (1- \ve) |V' \cap V_1| \cdot \overline{p}_1 ~\big|~ |V' \cap V_1| = s \right) &\leq \frac{1}{\ve^2 s \overline{p}_1}. \label{eq:Chebyshev-step}
\end{align}
Recall that $|V'| > \frac{3}{4}|V_1\cup V_2|$ with probability $1-o(1)$. $|V_1|=|V_2|=|V_1\cup V_2|/2$, so with high probability $|V_1\cap V'|>n^{1-\delta}$ and $|V_2\cap V'|>n^{1-\delta}$. Along with \eqref{eq:Chebyshev-step}, this implies
\begin{align*}
\mathbb{P}\left(Y_1 \leq (1-\ve) n^{1-\delta}\overline{p}_1 \right) &\leq \frac{n^{\delta-1}}{\ve^2 \overline{p}_1} + o(1).
\end{align*}
Since $\overline{p}_1\ge  n^{2\delta-1+o(1)}$, we have shown that there is a function $f(n) = \omega(1)$ such that $\mathbb{P}\left(Y_1 \leq f(n)\right) = o(1)$. By a similar argument, it holds that $\mathbb{P}\left(Y_2 \leq f(n)\right) = o(1)$. Applying Proposition \ref{proposition:equiprobable} with $\delta = 1 - o(1)$ and $k = \left(f(n)\right)^2$ completes the proof. 
\end{proof}

\subsection{Sharpness in a general parameter regime: Proof of Theorem \ref{theorem:sparse-case-sharpness}}
We start by setting up a coupling argument that will be used in the proof.
\begin{definition}[Reduced \CPDS~model]\label{defn:reduced-CPDS}
Given an instance of $\CPDS(p,q,\alpha,K)$ and $\alpha'\geq \alpha$, suppose we independently add an absent edge between each pair of vertices whose connection was previously censored with probability $\frac{\alpha'-\alpha}{1- \alpha}$.
We call this model a reduced \CPDS~model and denote it by $\CPDS'(p,q,\alpha,\alpha', K)$. 
\end{definition}
We observe the following: 
\begin{lemma}\label{lem:reduced-CPDS}
The distribution of $\CPDS'(p,q,\alpha,\alpha', K)$ is equal to that of $\CPDS(p\alpha/\alpha',q\alpha/\alpha',\alpha',K)$.
\end{lemma}
\begin{proof}
Clearly, the probability of present edges and censored edges are equal in both the models, and therefore, the probabilities of absent edges are also equal. 
\end{proof}
Since $\CPDS(p,q,\alpha,\alpha', K)$ is distributed as $\CPDS(p\alpha/\alpha',q\alpha/\alpha',\alpha',K)$, any algorithm on the latter model would result in an algorithm in the former model via the reduction. In particular, the following gives a spectral algorithm for $\CPDS(p,q,\alpha,K)$ via reduction:\\
\begin{breakablealgorithm}
\caption{Reduced spectral recovery in the CPDS model}\label{alg:CPDS-reduced}
\begin{algorithmic}[1]
\Require{Parameters $p$, $q$, $\alpha$, $\alpha'$ ($\alpha'\geq \alpha$) and $K$; an unlabeled observation  $A \sim \CPDS(p,q,\alpha,K)$.} \vspace{.2cm}
\Ensure{An estimate of the planted dense subgraph vertices.}
\State Create a reduced \CPDS~instance $A' \sim \CPDS'(p,q,\alpha,\alpha', K)$ according to Definition~\ref{defn:reduced-CPDS}. 
\State If $\alpha' = 1$, then apply Algorithm \ref{alg:PDS} to $A'$, obtaining an estimate $\hat{S}$. Otherwise, apply Algorithm~\ref{alg:CPDS} to $A'$.
\State Return $\hat{S}$.
\end{algorithmic}
\end{breakablealgorithm}
\vspace{.2cm}

\begin{proof}[Proof of Theorem~\ref{theorem:sparse-case-sharpness}]
Recall \eqref{eq:CH-defn}. Using $\alpha_n p_n = \frac{a \log n}{n}, \alpha_n q_n = \frac{b \log n}{n}$, and $\alpha_n = \frac{t_n \log n}{n}$ we can simplify
\begin{eq} 
    \frac{\alpha_n n}{\log n} \Delta_+(p_n,q_n) &= t_n \max_{x\in [0,1]} \left(1-p_n^{x}q_n^{1-x} - (1-p_n)^{x}(1-q_n)^{1-x} \right)\\
    & = \max_{x\in [0,1]} t_n  \bigg(1-\bigg(\frac{a}{t_n}\bigg)^{x}\bigg(\frac{b}{t_n}\bigg)^{1-x} - \bigg(1-\frac{a}{t_n}\bigg)^{x}\bigg(1-\frac{b}{t_n}\bigg)^{1-x} \bigg) \\
    & = \max_{x\in [0,1]}   \left(t_n- a^{x}b^{1-x} - (t_n-a)^{x}(t_n-b)^{1-x}\right)\\
    &:= \max_{x\in [0,1]} g(t_n,x), \label{eq:gt-n}
\end{eq}
where we have denoted the expression inside the max as $g(t_n,x)$. 
Considering $g$ as a function of $t$, note that $\frac{\dif g(t,x)}{\dif t} = 1-x(\frac{t-b}{t-a})^{1-x}-(1-x)(\frac{t-a}{t-b})^x\le 0$, by the weighted AM-GM inequality.
Thus, $g(t,x)$ is nonincreasing in $t$, and therefore, $\max_{x\in [0,1]}g(t,x)$ is also nonincreasing in $t$. 
Also, for $t,t'> \max\{a,b\}$, 
\begin{align*}
    |g(t,x) - g(t',x)| \leq |t-t'|+ (t-a)^x |(t-b)^{1-x}-(t'-b)^{1-x}|+(t'-b)^x |(t-a)^{x}-(t'-a)^{x}|. 
\end{align*}
Thus, for any $t_n\to t>\max\{a,b\}$, we have $\max_{x\in [0,1]}g(t_n,x) \to \max_{x\in [0,1]}g(t,x)$, and therefore,  $\max_{x\in [0,1]}g(t,x)$ is continuous in $t$ on the interval $(\max\{a,b\}, \infty)$. 
To see that $\max_{x \in [0,1]} g(t,x)$ is also continuous at $t = \max\{a,b\}$, let us assume without loss of generality that $a > b$. We will show that the $\limsup_{t\searrow}  \max_{x\in [0,1]}g(t,x)$ and $\liminf_{t\searrow}  \max_{x\in [0,1]}g(t,x) $ are equal. For any $t\geq a$
\begin{align*}
    \max_{x\in [0,1]} g(t,x) \leq \max_{x\in [0,1]} \left(t-a^{x}b^{1-x}\right) = t - b \min_{x\in [0,1]}\bigg(\frac{a}{b}\bigg)^{x} = t-b. 
\end{align*}
Therefore, $\limsup_{t\searrow a}\max_{x\in [0,1]} g(t,x) \leq a-b$. 
Also, for any fixed $\varepsilon>0$, 
\begin{align*}
    \max_{x\in [0,1]} g(t,x) \geq  t - a^\varepsilon b^{1-\varepsilon} - (t-a)^{\varepsilon}(t-b)^{1-\varepsilon},
\end{align*}
and therefore $\liminf_{t\searrow a}\max_{x\in [0,1]} g(t,x) \geq a- b \big(\frac{a}{b}\big)^{\varepsilon}$. 
Taking $\varepsilon\to 0$, we can conclude that $\liminf_{t\searrow a}\max_{x\in [0,1]} g(t,x) \geq a- b $, which proves the continuity of $\max_{x\in [0,1]} g(t,x)$ at $t = a$.

Next, by Lemma \ref{lemma:connected-conditions}, 
\begin{equation}
\lim_{t\to\infty} \max_{x\in [0,1]}g(t,x)= f(a,b). \label{eq:f(a,b)-limit}
\end{equation}
To prove the first statement in Theorem~\ref{theorem:sparse-case-sharpness}, let $ \liminf_{n\to\infty} \frac{\alpha_n n}{\log n} \times \rho \Delta_+(p_n,q_n) >  1.$ 
We consider the cases where $\rho f(a,b) >1$ and $\rho f(a,b) \leq 1$ separately.

If $\rho f(a,b)>1$ then by Theorem \ref{theorem:dense-subgraph-1} there is a spectral algorithm that recovers the dense subgraph of $\CPDS(a\log(n)/n,b\log(n)/n,1,K_n)$ with probability $1-o(1)$. By Lemma~\ref{lem:reduced-CPDS}, the reduced spectral algorithm in Algorithm~\ref{alg:CPDS-reduced} recovers the dense subgraph of $\CPDS(a\log(n)/(\alpha_n n),b\log(n)/(\alpha_n n),\alpha_n,K_n) = \CPDS(p_n, q_n, \alpha_n, K_n)$.  

If $\rho f(a,b)\le 1$ but $ \liminf_{n\to\infty} \frac{\alpha_n n}{\log n} \times \rho \Delta_+(p_n,q_n) >  1$, then we claim that there must exist some fixed $t_0>\max\{a,b\}$ such that 
\begin{eq} \label{alpha-upper-bound}
 t_0\times \rho \Delta_+\left(\frac{a}{t_0},\frac{b}{t_0}\right) >1, \quad\text{and } \alpha_n\leq \frac{t_0\log(n)}{n} \quad \text{for all sufficiently large }n. 
\end{eq}
To prove \eqref{alpha-upper-bound}, recall $t_n = \alpha_n n /\log n$ 
and let $t_0' = \limsup_{n\to\infty} t_n$. 
Let us first show that $t_0'<\infty$. Indeed, if $t_n = \alpha_n n /\log n$ and $t_{n_k} \to \infty$ along some subsequence $(n_k)_{k\geq 1}$, then by \eqref{eq:gt-n} and \eqref{eq:f(a,b)-limit}, we must have that 
\begin{align*}
    \lim_{k\to\infty} \frac{\alpha_{n_k} n_k}{\log n_k} \times \rho \Delta_+(p_{n_k},q_{n_k}) = \rho f(a,b).
\end{align*}
This would yield a contradiction. To prove the first assertion in \eqref{alpha-upper-bound},
let $(n_k')_{k\geq 1}$ be a subsequence such that $t_{n_k'} \to t_0'$. 
We have shown above that $t \Delta_+\big(\frac{a}{t},\frac{b}{t}\big) = \max_{x\in [0,1]} g(t,x)$ is continuous for $t\geq \max\{a,b\}$.  
Thus
\begin{align*}
     1<\lim_{k\to\infty} \frac{\alpha_{n_k'} n_k'}{\log n_k'} \times \rho \Delta_+(p_{n_k'},q_{n_k'}) = t_0'\times \rho \Delta_+\Big(\frac{a}{t_0'},\frac{b}{t_0'}\Big).
\end{align*}
Using the continuity of $t \Delta_+\big(\frac{a}{t},\frac{b}{t}\big)$ again, we can pick $t_0>t_0'$ such that both the conditions in \eqref{alpha-upper-bound} are ensured. 
Thus, we assume \eqref{alpha-upper-bound} holds. Due to the first condition in  \eqref{alpha-upper-bound},  there is a spectral algorithm that recovers the dense subgraph of $\CPDS(a/t_0,b/t_0,t_0\log(n)/n,K_n)$ by Theorem \ref{theorem:exact-recovery-censored-1}. The reduction in Definition~\ref{defn:reduced-CPDS}, together with the second condition of \eqref{alpha-upper-bound} and  Lemma~\ref{lem:reduced-CPDS}, implies that there is also a reduced spectral algorithm (Algorithm~\ref{alg:CPDS-reduced}) that recovers the dense subgraph of \[\CPDS(a\log(n)/(\alpha_n n),b\log(n)/(\alpha_n n),\alpha_n,K_n) = \CPDS(p_n, q_n, \alpha_n, K_n).\] 

To prove the second statement in Theorem~\ref{theorem:sparse-case-sharpness}, let $ \limsup_{n\to\infty} \frac{\alpha_n n}{\log n} \times \rho \Delta_+(p_n,q_n) < 1$. We can set $\alpha_n'=\min(\alpha_n, \log^2(n)/n)$. 
In that case, 
\begin{eq} \label{asymp-t-t-prime}
\lim_{n\to\infty} \left|\frac{\alpha_n n}{\log n} \times \rho \Delta_+(p_n,q_n)-\frac{\alpha_n' n}{\log n} \times \rho \Delta_+(p_n\alpha_n/\alpha_n',q_n\alpha_n/\alpha_n')\right|=0.
\end{eq}
To see this, let  $t_n'  = \frac{\alpha_n' n}{\log n}$. 
We will use the elementary analysis fact that a sequence $(y_n)_{n\geq 1}$ converges to $y$ if and only if for any subsequence there is a further subsequence that converges to $y$.
Fix any subsequence $(n_k)_{k\geq 1}$. 
If $\limsup_{k\to\infty} t_{n_k} <\infty$, then $t_{n_k} = t_{n_k}'$ and the limit in \eqref{asymp-t-t-prime} holds trivially along $(n_k)_{k\geq 1}$. 
If $\limsup_{k\to\infty} t_{n_k} =\infty$, there is a further subsequence $(m_k)_{k\geq 1}\subset (n_k)_{k\geq 1}$ such that $\lim_{k\to\infty} t_{m_k} =\infty$. In this case, $t_{m_k}' \leq t_{m_k}$ and also $\lim_{k\to\infty} t_{m_k}' =\infty$. 
Since we have shown below \eqref{eq:gt-n} that $t\Delta_+(a/t,b/t)$ is non-increasing in $t$, we  have
\[\lim_{k\to\infty} \Big|t_{m_k} \times \rho \Delta_+(a/t_{m_k},b/t_{m_k})-t_{m_k}' \times \rho \Delta_+(a/t_{m_k}',b/t_{m_k}')\Big|=0,\] and \eqref{asymp-t-t-prime} holds along $(m_k)_{k\geq 1}$. This proves \eqref{asymp-t-t-prime}.
Thus, by Theorem \ref{theorem:impossibility2} there is no algorithm that recovers the dense subgraph of $\CPDS(a\log(n)/(\alpha_n' n),b\log(n)/(\alpha_n' n),\alpha_n',K_n)$ with nonvanishing probability. By the reduction in Definition~\ref{defn:reduced-CPDS} and Lemma~\ref{lem:reduced-CPDS}, there is no algorithm that recovers the dense subgraph of $\CPDS(p_n,q_n,\alpha_n,K_n)$ with nonvanishing probability.

\end{proof}

\section{Spectral submatrix localization} \label{sec:Gaussian}
In this section, we will complete the proof of Theorem~\ref{theorem:submatrix-localization-2}. 
We start by recalling the following result by Hajek, Wu, and Xu \cite{Hajek2017}.
\begin{theorem}[Corollary 4, \cite{Hajek2017}]
Recall the definition of exact recovery from \eqref{exact-recovery-possible}. If 
\begin{align}
&\lim_{n \to \infty} K \mu^2 = \infty \label{eq:SL-condition-1}\\
&\liminf_{n \to \infty} \frac{(K-1) \mu^2}{\log \frac{n}{K}} > 4 \label{eq:SL-condition-2}\\
&\liminf_{n \to \infty} \frac{K \mu^2}{\left(\sqrt{2 \log n} + \sqrt{2 \log K} \right)^2} > 1, \label{eq:SL-condition-3}
\end{align}
then exact recovery is possible. If exact recovery is possible, then \eqref{eq:SL-condition-1} holds, while \eqref{eq:SL-condition-2} and \eqref{eq:SL-condition-3} hold with non-strict inequality.
\end{theorem}
Note that \eqref{eq:SL-condition-3} implies \eqref{eq:SL-condition-1}. We consider the regime where $K= \lfloor \rho n \rfloor$ for a constant $\rho = (0,1)$ and $\mu = a\sqrt{\frac{\log(n)}{n}}$ 
for a constant $a > 0$. In this regime,  \eqref{eq:SL-condition-1} implies \eqref{eq:SL-condition-2}. Therefore, we have the following recovery condition.

\begin{corollary}\label{corollary:Gaussian-threshold}
Let $\rho \in (0,1)$ and $a >0$ be constants. Let $K = \lfloor \rho n\rfloor$ and $\mu = a \sqrt{\frac{\log n}{n}}$, and let $A \sim \SL(\mu, K)$. If $\rho a^2 > 8$,
then exact recovery is possible. If $\rho a^2 < 8$
then exact recovery is impossible.
\end{corollary}
\begin{proof}
The result follows from examining \eqref{eq:SL-condition-3}:
\begin{align*}
\liminf_{n \to \infty} \frac{K \mu^2}{\left(\sqrt{2 \log n} + \sqrt{2 \log K} \right)^2} &= \liminf_{n \to \infty}\frac{\rho a^2 \log (n)}{\left(\sqrt{2 \log n} + \sqrt{2 \log (\rho n)} \right)^2}\\
&= \frac{\rho a^2}{(2 \sqrt{2})^2}= \frac{\rho a^2}{8}.
\end{align*}
\end{proof}
In order to prove Theorem \ref{theorem:submatrix-localization}, we first prove an entrywise result akin to Lemma \ref{lemma:entrywise}.

\begin{lemma}\label{lemma:entrywise-Gaussian}
Let $a >0$ and $\rho \in (0,1)$ be constants. Let $A \sim \SL(\mu, K)$, where $K = \lfloor \rho n \rfloor$ and $\mu = a \sqrt{\frac{\log n}{n}}$. Let $(\lambda, u)$ be the top eigenpair of $A$. Similarly, let $(\lambda^{\star}, u^{\star})$ be the top eigenpair of $A^{\star} = \mathbb{E}[A |\true]$. Then with probability $1 - o(1)$,
\[\min_{s \in \{\pm1\}} \left\Vert su - \frac{A u^{\star}}{\lambda^{\star}} \right \Vert_{\infty} \leq \frac{C}{\sqrt{n \log n}},\]
where $C = C(\rho, a) >0$ is a constant.
\end{lemma}

\begin{proof}
As in the proof of Lemma \ref{lemma:entrywise}, we verify the assumptions of \cite[Theorem 2.1]{Abbe2020}, using ideas from the proof of \cite[Theorem 3.1]{Abbe2020}. Note that $A^{\star}$ is a rank-$1$ matrix, where $\lambda^{\star} = \mu K$. Therefore $\Delta^{\star} = \lambda^{\star}$ and $\kappa = 1$.

We have
\[\Vert A^{\star} \Vert_{2 \to \infty} = \sqrt{K \mu^2} = \mu \sqrt{K} = \Theta(\sqrt{\log n}).\]
Let $\gamma = \frac{3\sqrt{n}}{\mu K}$, so that $\Vert A^{\star} \Vert_{2 \to \infty} \leq \gamma \lambda^{\star} = \Theta(\sqrt{n})$, verifying the first assumption.

The second assumption is trivially satisfied.

To verify the third assumption, we apply \cite[Theorem 2.11]{Davidson2001}, which implies that for $t > 0$,
\[\mathbb{P}\left(\Vert A - A^{\star}\Vert_2 \geq 2 \sqrt{n} + t \right) \leq 2 e^{-t^2/4}\]
(see also \cite[Lemma 25]{Hajek2016-hidden-community}). We set $t = \sqrt{n}$, so that $2 \sqrt{n} + t = \gamma \Delta^{\star}$, allowing us to take $\delta_0 = 2e^{-n/4}$. Set $\varphi(x) = cx$ for a constant $c > 0$ to be determined, so that 
\[32 \kappa \max\{\gamma, \varphi(\gamma)\} = 32\max(1,c) \gamma = o(1).\]

To verify the fourth assumption, fix $w \in \mathbb{R}^n$. Note that 
\begin{align*}
\Delta^{\star} \Vert w \Vert_{\infty} \varphi\left(\frac{\Vert w \Vert_2}{\sqrt{n} \Vert w \Vert_{\infty}} \right) &=  \frac{c\Delta^{\star} \Vert w \Vert_2}{\sqrt{n}}\\
&=\frac{ca\sqrt{\log n} K \Vert w \Vert_2}{n}\\
&\geq (1-\epsilon)c\rho a \sqrt{\log n}  \Vert w \Vert_2
\end{align*}
for any $\epsilon$, for $n$ sufficiently large.
For each $m \in [n]$, we have $(A - A^{\star})_{m, \cdot} w \sim \textsc{Normal}(0, \Vert w \Vert_2^2)$. Let $Z \sim \textsc{Normal}(0,1)$. Then
\begin{align*}
\mathbb{P}\left(\left|(A - A^{\star})_{m, \cdot} w \right| \leq \Delta^{\star} \Vert w \Vert_{\infty} \varphi\left(\frac{\Vert w \Vert_2}{\sqrt{n} \Vert w \Vert_{\infty}} \right) \right) &\geq \mathbb{P}\left(\left|(A - A^{\star})_{m, \cdot} w \right| \leq (1-\epsilon)c\rho a \sqrt{\log n}  \Vert w \Vert_2 \right)\\
&= \mathbb{P}\left(Z \Vert w \Vert_2  \leq (1-\epsilon)c\rho a \sqrt{\log n}  \Vert w \Vert_2 \right)\\
&\geq 1 - \frac{2}{(1-\epsilon)c\rho a\sqrt{2\pi \log(n)}} e^{-(1-\epsilon)^2c^2\rho^2 a^2 \log(n)/2}\\
&= 1 - \frac{2}{(1-\epsilon)c\rho a\sqrt{2\pi \log(n)}} n^{-(1-\epsilon)^2c^2\rho^2 a^2 /2}
\end{align*}
In the second inequality, we have used the fact that $\mathbb{P}(Z > t) \leq \frac{1}{t\sqrt{2\pi}} e^{-t^2/2}$.
Therefore, we may take \[\delta_1 = \frac{2}{(1-\epsilon)c\rho a\sqrt{2\pi \log(n)}} n^{1-(1-\epsilon)^2c^2\rho^2 a^2 /2}.\]
Set $c > \frac{\sqrt{2}}{(1-\epsilon)\rho a}$ to ensure that $\delta_1 = o(1)$.

Finally, since $\gamma = \Theta\left(\frac{1}{\sqrt{\log n}}\right)$, $\Vert A^{\star} \Vert_{2 \to \infty} = \Theta(\sqrt{\log n})$, and $\Delta^{\star} = \Theta(\sqrt{n \log n})$,
we have
\begin{align*}
\kappa (\kappa + \varphi(1)) (\gamma + \varphi(\gamma)) \Vert u^{\star}\Vert_{\infty} + \frac{\gamma \Vert A^{\star} \Vert_{2 \to \infty}}{\Delta^{\star}} &= \Theta\left(\frac{1}{\sqrt{n \log n}} \right)
\end{align*}
The desired conclusion then follows from \cite[Theorem 2.1]{Abbe2020}.

\end{proof}

We also need a Gaussian concentration result, such as the following variant of \cite[Lemma 2]{Hajek2016-hidden-community}.
\begin{lemma}\label{lemma:Gaussian-concentration}
Let $\{Z_i\}_n$ be a sequence of (not necessarily independent) normal random variables, where $Z_i \sim \textsc{Normal}(0,1)$. Then
\[\max_{i \in [n]} Z_i \leq \sqrt{2\log n}\]
with probability $1-o(1)$.
\end{lemma}

\begin{proof}
\begin{align*}
    \PR\left(\max_{i\in [n]} Z_i>\sqrt{2\log(n)}\right)&\le\sum_{i\in[n]} \PR\left(Z_i>\sqrt{2\log(n)}\right)\\
    &= n\int_{\sqrt{2\log(n)}}^\infty \frac{1}{\sqrt{2\pi}}e^{-x^2/2} dx\\
    &= \frac{n}{\sqrt{2\pi}}\int_{\sqrt{2\log(n)}}^\infty e^{-\log(n)-\sqrt{2\log(n)}\left(x-\sqrt{2\log(n)}\right)-\left(x-\sqrt{2\log(n)}\right)^2/2} dx\\
    &\le \frac{1}{\sqrt{2\pi}}\int_{\sqrt{2\log(n)}}^{\infty} e^{-\sqrt{2\log(n)}\left(x-\sqrt{2\log(n)}\right)} dx\\
    &= \frac{1}{\sqrt{4\pi\log(n)}}.
\end{align*}
\end{proof}

\begin{proof}[Proof of Theorem \ref{theorem:submatrix-localization}]
Corollary \ref{corollary:Gaussian-threshold} implies that the MAP estimator succeeds in recovering $\true$ with high probability. Therefore, it suffices to show that thresholding either $u$ or $-u$ successfully recovers $\true$, with high probability.

Applying Lemma \ref{lemma:entrywise-Gaussian}, let $s^{\star} \in \{\pm 1\}$ be such that
\[ \left\Vert s^{\star} u - \frac{A u^{\star}}{\lambda^{\star}} \right \Vert_{\infty} \leq \frac{C}{\sqrt{n \log n}}.\]
This occurs with probability $1-o(1)$.

For $i \in [n]$,
\begin{align*}
\left(\frac{A u^{\star}}{\lambda^{\star}} \right)_i &= \frac{1}{\mu K^{3/2}}\sum_{j \in \true} A_{ij}
\sim \begin{cases}
N\left(\frac{1}{\sqrt{K}}, \frac{1}{\mu^2 K^2}\right) & i \in \true\\
N\left(0, \frac{1}{\mu^2 K^2}\right) & i \not \in \true.
\end{cases}
\end{align*}
By Lemma \ref{lemma:Gaussian-concentration}, 
\[\max_{i \not \in S^{\star}} \mu K \left(\frac{A u^{\star}}{\lambda^{\star}} \right)_i \leq \sqrt{2 \log (n-K)}\]
with high probability. 
Similarly, 
\[\min_{i \in S^{\star}} \mu K\left[ \left(\frac{A u^{\star}}{\lambda^{\star}} \right)_i - \frac{1}{\sqrt{K}}\right] \geq -\sqrt{2 \log K}\]
with high probability.
Combining these facts along with the entrywise bound, we conclude that with high probability,
\begin{align*}
\min_{i \in S^{\star}} s^{\star}u_i &\geq \frac{1}{\sqrt{K}} - \frac{\sqrt{2  \log K}}{\mu K} - \frac{C}{\sqrt{n \log n}}\\
\max_{i \not \in S^{\star}} s^{\star}u_i &\leq \frac{\sqrt{2  \log (n-K)}}{\mu K} + \frac{C}{\sqrt{n \log n}}.
\end{align*}
We see that
\begin{align*}
\min_{i \in S^{\star}} s^{\star}u_i - \max_{i \not \in S^{\star}} s^{\star}u_i &\geq \frac{1}{\sqrt{K}} - 2\frac{\sqrt{2  \log n}}{\mu K} - O\left(\frac{1}{\sqrt{n \log(n)}} \right)\\
&= \frac{1}{\sqrt{K}} \left(1  - \sqrt{\frac{8 \log n}{\mu^2 K}} \right) - O\left(\frac{1}{\sqrt{n \log(n)}} \right)\\
&= \frac{1}{\sqrt{K}} \left(1  - \sqrt{\frac{8n \log n}{a^2 \lfloor \rho n\rfloor \log n }} \right) - O\left(\frac{1}{\sqrt{n \log(n)}} \right)\\
&= \frac{1}{\sqrt{K}} \left(1  - (1+o(1))\sqrt{\frac{8 }{\rho a^2}} \right) - O\left(\frac{1}{\sqrt{n \log(n)}} \right).
\end{align*}
Since $\rho a^2 > 8$ and $K  = \Theta(n)$, we have
\[\min_{i \in S^{\star}} s^{\star}u_i - \max_{i \not \in S^{\star}} s^{\star}u_i > 0\]
with high probability.
We conclude that thresholding $s^{\star} u$ succeeds in recovering the communities with probability $1-o(1)$. 
\end{proof}
\bibliography{../references}
\bibliographystyle{abbrv}

\end{document}